\documentclass[12pt,notitlepage]{report}
\usepackage[margin=0.7in]{geometry}

\usepackage{amssymb}
\usepackage{amsmath}
\usepackage[perpage,para,symbol*]{footmisc}

\usepackage[boxruled,algochapter,commentsnumbered]{algorithm2e}
\SetAlgoNoEnd

\usepackage{graphicx} 
\usepackage{pgf,pgfarrows,pgfnodes}

\newtheorem{thm}{Theorem}[chapter]
\newtheorem{defn}{Definition}[chapter]
\newtheorem{obsrv}{Observation}[chapter]
\newtheorem{clm}{Claim}[chapter]
\newtheorem{lemma}{Lemma}[chapter]
\newtheorem{corollary}{Corollary}[chapter]
\newtheorem{prop}{Proposition}[chapter]

\newtheorem{remark}{Remark}[chapter]

\newcommand{\es}{\emptyset}
\newcommand{\e}{\varepsilon}
\newcommand{\bb}[1]{\mathbb{#1}}
\newcommand{\cl}[1]{\mathcal{#1}}
\newcommand{\Ftwo}{\mathbb{F}_{2}}
\newcommand{\Fq}{\mathbb{F}_{q}}
\newcommand{\zo}[1]{\{0,1\}^{#1}}
\newcommand{\ga}{\alpha}
\newcommand{\gb}{\beta}

\newcommand{\sus}{\subseteq}

\newenvironment{proof}[1][Proof]{\par \textbf{#1.} }{\hspace{10pt}\hfill$\blacksquare$\par}
\title{Truth Table Minimization of Computational Models\footnote{This work was done under the supervision of Prof. Eyal Kushilevitz as a partial fulfilment of the requirements for the degree of Master of Science in Computer Science in the Technion, Haifa 32000, Israel.}}
\author{Netanel Raviv\footnote{Department of Computer Science, Technion---Israel Institute of Technology,  Technion City, Haifa 32000, Israel}
}
\date{\today}
\begin{document}

\maketitle
\begin{abstract}
Complexity theory offers a variety of concise computational models for computing boolean functions - branching
programs, circuits, decision trees and ordered binary decision diagrams to name a few. A natural question that arises in this context with respect to any such model is this:
\begin{center}
Given a function $f:\zo{n} \to \zo{}$, can we compute the optimal complexity of computing $f$ in the
computational model in question? (according to some desirable measure).
\end{center}

A critical issue regarding this question is how exactly is $f$ given, since a more elaborate description of $f$ allows the algorithm to use more computational resources. Among the possible representations are black-box access to $f$ (such as in computational learning theory), a representation of $f$ in the desired computational model or a representation of $f$ in some other model. One might conjecture that if $f$ is given as its complete truth table (i.e., a list of $f$'s values on each of its $2^n$ possible inputs), the most elaborate description conceivable, then any computational model can be efficiently computed, since the algorithm computing it can run $poly(2^n)$ time. Several recent studies show that this is far from the truth - some models have efficient and simple algorithms that yield the desired result, others are believed to be hard, and for some models this problem remains open. In this thesis we will discuss the computational complexity of this question regarding several common types of computational models. We present several new hardness results and efficient algorithms, as well as new proofs and extensions for known theorems, for variants of decision trees, formulas and branching programs.

\end{abstract}
\tableofcontents

\chapter{Introduction} \label{chapt:Intro}
A classic question in computer science, which dates back to the early 60s, is how can we define and compute the ``complexity'' of a given string. The question is natural - the string ``0101010'' is intuitively simpler than the string ``9846723''. This notion is called \textit{Kolmogorov Complexity} and it is denoted by $K(x)$ in the literature (for more information on Kolmogorov complexity, see \cite{LV97}). Consider the following definition: a complexity of a string $s$ is the length of the shortest computer program that outputs $s$ when executed on the empty input. It is not hard to imagine situations where such program is sought, and why this definition is useful. Moreover, one may see why this would be a natural definition for string complexity - ``simple'' strings are the ones that present some coherent structure, and are printable by ``short'' computer programs, while ``complicated'' strings are the ones that do not exhibit such a structure thus requiring long computer programs for that purpose. \\
\indent Formally, the complexity $K(x)$ of a string $x$ is defined as the minimum number of states of a Turing machine that prints $x$ on the empty input. The function $K(x)$ was defined in \cite{K98}, where it was proved to be a non-computable function. Seeking to expand this line of research, one might ask similar questions about the complexity of printing strings using different computational models for model  that received a wide attention in modern complexity theory during the past couple of decades.\\
\indent The Turing machine model, albeit being simple, provides extraordinary challenges when trying to prove hardness results (i.e., showing that a certain problem cannot be solved efficiently). As a consequence, simpler models were offered, whose outputs are (in most cases) a single bit. Our concern in this thesis is the question of computing the complexity of strings in some of those models. This may seem insightful since it may provide us with a better understanding of the complexity of strings when the computation methods at hand are limited somehow. In order to do so, we ought to alter our view of ``string complexity'' to comply with computational models other than the Turing machine. Therefore, an alternative point of view of the question of string complexity is offered: we require that the computational model of interest will provide us with the $i$th bit of the string $x$ when given a binary representation of $i$, rather than the entire string $x$. We observe that this new definition for string complexity is computable for Turing machines iff $K(x)$ is computable. Moreover, these definitions differ by at most some multiplicative factor. This new requirement views the string $x$ as a function $f:\zo{\lceil \log x \rceil} \to \zo{}$, where $x$ serves as the truth table of $f$. In light of this point of view of string complexity, the main question that we ask in this thesis is: Given a full truth table of a function $f:\zo{n} \to \zo{}$, can we compute the optimal complexity of computing $f$ in the computational model in question? (where optimality is according to some desirable measure). Henceforth, when discussing a certain  computational model $C$, we shall name this question "the truth-table minimization problem of $C$".\\
\indent The computational models we are interested in were developed in different contexts and under different motivations. All those models supply some insights about boolean functions computation, which are of independent interest, but some may also shed light on different aspects of computer science and computer engineering. E.g., the boolean circuit and the formula models (see Chapter ~\ref{chapt:Formulas}), which resemble real-life digital circuits, are studied in both computational complexity, since they supply a simplified approach towards understanding parallel computing and polynomial Turing machines, and in hardware design. Being very difficult to understand themselves, it is common to restrict the structure of the circuits or formulas at hand in order to achieve better understanding of their power and limitations. Common restrictions are depth restriction, limited types of gates, limited number of appearances of every variable, etc. Branching programs (see Chapter ~\ref{chapt:BPs}) were studied in complexity theory mainly because they constitute an automaton-like modelling for bounded space algorithms. Very much like circuits, certain restrictions (e.g., number of appearances of every variable) are often applied on branching programs in order to understand them better. A common restriction is setting the variables in a pre-defined order, which result in a model called OBDD (ordered binary decision diagrams, see also in Chapter \ref{chapt:BPs}). Decision trees (Chapter \ref{chapt:DT}), which are abundant also outside of computer science (e.g., in medical diagnosis or in risk management), provide a strong insight of the inherent complexity of computing a function in a simple if-else environment using a very small set of possible conditions.\\
\indent Most truth table minimization problems (for the aforementioned models) were studied in the literature. Our work may be seen as a direct continuation of some of the works mentioned below. For the decision tree model, \cite{GLR99} provide a very simple dynamic programming algorithm for truth-table minimization. In Chapter \ref{chapt:DT} we preform some modifications to their algorithm in order to achieve efficient (polynomial or quasi-polynomial) algorithms for two natural variants of the decision tree model: linear decision trees (where nodes contain linear functions in arbitrarily many variables, see Section \ref{LDT}) and read once decision trees with symmetric functions in nodes (where every variable appears at most once in every path from the root to a leaf, and all nodes contain some symmetric function between arbitrarily many variables, see Section \ref{sctn:SymmetricTrees}). Moreover, we provide an NP-completeness result for another variant of the decision tree model (Section \ref{sctn:TreesNPC}), in which the set of possible tests in the nodes is given as input, together with the full truth table. A similar problem in the non-boolean world was asked by \cite{HR76} where it was proved to be NP-complete. In the branching program model, we give a simpler proof for the hardness-of-approximation result by \cite{AKRR03} using a generalization of a method used by \cite{KW09, AH+06} (presented in Chapter \ref{chapt:Pseudo}). Namely, we give a different proof for the fact that truth table minimization of branching programs is inapproximable up to a factor of $2^{cn}$, where $n$ is the number of variables, and for every $c\in (0\frac{1}{2})$. We also give an algorithm for $\mu$-branching program truth-table minimization (where every variable may appear only once in the entire program, see Definition \ref{defn:OBDD}) which is faster than applying the OBDD (see Definiton \ref{defn:OBDD}) algorithm of \cite{FS90}. \\
\indent In another model, \cite{KC00} address the truth-table minimization problem of boolean circuits (MCSP). While they were unable to provide a definite classification of the decisional variant of this problem (P or NP-complete), they gave some strong evidence of what might be the correct answer, and what its implications would be. Namely, they showed that if MCSP is in $P$, then there are no pseudorandom functions in $P/poly$ (a result that would undermine almost all modern cryptography). We show a similar result regarding any computational model in Chapter \ref{chapt:Pseudo}, a result which also generalizes several works regarding the truth table minimization of $AC^0$ circuits \cite{AH+06} and communication complexity \cite{KW09}. In addition, \cite{KC00} also show that if MCSP is NP-complete under a reduction that is ``natural'' in some sense, then this would immediately imply an explicit construction of a function with high circuit complexity, which is a long-standing open problem. We use a similar technique, that together with Valiant's depth reduction lemma (see \cite{V77, V09}) provides an explicit construction of a different kind of hard functions (which is also an open problem), assuming the natural NP-completeness of a seemingly easier problem - the truth-table minimization of depth-3 formulas (see Section \ref{sctn:LowerBoundsSigma3}). \\
\indent In Chapter \ref{chapt:Formulas} we provide efficient algorithms for the construction of several types of read-once formulas (a formula with at most one appearance of every variable). The main theorem in this chapter regards the uniqueness of several types of decompositions of a boolean function (i.e., representation of the function as $\wedge$, $\vee$ or $\oplus$ of variable disjoint factors, see Corollary \ref{cor:Main}). The proof of this theorem strongly relies on several theorems regarding partial derivatives of multilinear polynomials that were developed by Shpilka and Volkovich in \cite{SV08, SV10, V12}. While similar work regarding read-once formulas was already done (\cite{GAU04} showed a recognition algorithm for read-once formulas given a DNF representation, while \cite{P95} gave a construction algorithm for read-once formulas given a DNF that uses graph-theoretic tools), our work uses the aforementioned algebraic tools in order to construct several variants and generalizations of the read-once formula model. For example, a larger set of possible gates (Section ~\ref{sctn:ROFXOR}), larger readability (Sections ~\ref{sctn:UF2} and \ref{sctn:F2A}) and costly negation gates (Sections ~\ref{sctn:CostlyBoolean} and \ref{sctn:CostlyXOR}). In the same chapter we address the question of truth-table minimization of monotone depth-3 formulas. Being unsuccessful in providing a definite answer, we show some evidence of the hardness of this problem - we show that an algorithm for truth-table minimization of monotone formulas of depth 3, if exists, is unlikely to work in a \textit{serial} manner, i.e., to construct the minimal second level formulas one by one, since in this way it is most likely to encounter a problem which we prove to be $NP$-complete: the problem of finding the minimal monotone DNF for a partial truth table (Section \ref{sctn:monDNF}). Moreover, in Section \ref{sctn:FormulaDepth} we use a result by \cite{AKRR03} about inapproximabiltiy of formula size (under some cryptographic assumption) to show that the minimal formula depth is also inapproximable.\\
\indent We have also found a simple connection between the problems of truth-table minimization, learning and model minimization (Section \ref{sctn:TTMandLearning}), allowing us to use algorithms from learning theory to obtain truth-table minimization algorithms for several computational models (see Sections \ref{sctn:LDL}, \ref{sctn:OurROF} and \ref{sctn:BPEfficient}). 

\begin{remark} \label{remark:AllVariables}
All along this thesis we shall assume w.l.o.g that all given truth-tables represent functions which depend on all variables. This assumption does not limit the generality of the discussion since given a truth table of size $2^n$, we may run a simple $O(n\cdot 2^n)$ algorithm that verifies that the function indeed depends on all its variables, and if not it produces a truth-table of an equivalent sub-function that does depend on all its variables.
\end{remark}
\paragraph{Organization:} Chapters \ref{chapt:DT}, \ref{chapt:Formulas} and \ref{chapt:BPs} each contains a self-contained discussion about truth-table minimization of decision trees, formulas and branching programs respectively. Before presenting our results, each chapter begins by formally defining the model and its variants and summarizing known results. Chapter \ref{chapt:Pseudo} presents a generalization of a hardness result for truth-table minimization which appeared in several different papers in recent years. Some chapters are concluded with a discussion on open problems and further research directions.

\section{Truth-Table Minimization and Learning} \label{sctn:TTMandLearning}
Truth-table minimization's more popular counterpart, often named ``model minimization'', is the following problem: given \textit{some} model (e.g., a branching program, a decision tree, etc.) which represents a function $f$, can we efficiently find the minimal model which is consistent with $f$? This type of questions was a topic for extensive research throughout the years (e.g., \cite{ZB99,Sieling08} for decision trees, \cite{BW96} for OBDDs, \cite{BU08} for formulas or \cite{GD99} for branching programs). The reader may wonder whether there is a connection between the former and the latter problems. In addition, a reader which is familiar with computational learning theory may notice the resemblance between truth-table minimization and learning. In both we are given some kind of access to the values of a given function, and we are asked to decide if the function has a certain property. As we shall see, this intuition may be formalized, and some results may be deduced from it (see Sections \ref{sctn:LDL}, \ref{sctn:OurROF}, \ref{sctn:BPEfficient}). Since results about learning computational models are abound, this direction may lead to further lucrative research beyond the scope of this thesis (e.g., consider different models of learning and their implications to truth-table minimization, such as PAC learning). \\
\indent The learning model we consider consists of an algorithm with an oracle access to a function in one or more of the conventional ways, as in \cite{KM93, AHK89, BHH92, BBTV96, HSW90, BTW96, VW93, A87}. E.g.,

\begin{itemize}
\item A \textit{membership query}, where the algorithm supplies $x\in \{0,1\}^{n}$ and the oracle answers with $f(x)$.
\item  An \textit{equivalence query}, where the algorithm supplies a hypothesis $h\in C$ (where $C$ is the class of models to be learned) and the oracle answers either by saying ''yes'' or by supplying an $x$ such that $h(x)\ne f(x)$.
\item  A \textit{relevent possibility oracle}, where the algorithm specifies a set of literals, and the oracle answers if it is a subset of some minterm of $f$ (first defined in \cite{Valiant84}).
\end{itemize}

The output of the algorithm is a hypothesis which is identical or similar to the target function. The complexity of such algorithm is measured as a function of the number of variables of the target function. We call a learning algorithm \textit{exact} if the resulted hypothesis matches the target concept on all assignments. It will be called \textit{proper} if the resulted hypothesis is represented as a member of the class at question. Notice that in all of the papers mentioned above, the target function is guaranteed to reside within the concept class.

The following claim shows a connection between truth-table minimization, proper and exact learning and model minimization. 

\begin{thm} \label{thm:Learning}
Let $C$ be some class of models such that:
\begin{enumerate}
\item $C$ has a proper and exact learning algorithm $A$ running in $2^{O(n)}$ time (when $n$ if the number of variables of the target function); and
\item $C$ has an algorithm $B$ that receives a concept $c\in C$ and outputs a minimal equivalent $c'\in C$, and runs in time $2^{O(n)}$.
\end{enumerate}
Then $C$ has a polynomial truth-table minimization algorithm. 
\end{thm}
\begin{proof}
Observe that all mentioned types of queries may be simulated in polynomial (in $2^{O(n)}$) time when the full truth-table is given. E.g.:

\begin{itemize}
\item A \textit{membership query} is simply implemented by a truth-table look-up.
\item  An \textit{equivalence query} may be simulated by traversing all $2^{n}$ assignments and searching for a mismatch between $h$ and $f$. As all models considered in this thesis may be evaluated on a given input in $2^{O(n)}$ time, simulating an equivalence query may also be done in $2^{O(n)}$ time.
\item  A \textit{relevant possibility oracle} may also be simulated in $2^{O(n)}$ time. Since there are $3^n$ possible minterms, we may traverse them all and check if any of them constitutes a minterm by traversing all assignments.
\end{itemize}

Therefore we may define the following truth-table minimization algorithm: simulate $A$, answering all queries using the truth-table of $f$, and feed the result to $B$. However, in the usual settings in learning theory, the target function is promised to reside in the concept class, unlike in our settings. Therefore, we must add the following restrictions to the above algorithm: if $T(n)$ is a worst-case bound on the running time of $A$, we ought not to let $A$ run more than that much time. Moreover, we must make sure that the output of $A$ indeed represents $f$ before feeding it into $B$. This will assure us that the given truth-table truly lies in the concept class $C$, and if one of these conditions is not met, the correctness of $A$ allows us to deduce that the given truth-table does not have a representation in $C$, and we may reject.
\end{proof}

\chapter{Decision Trees} \label{chapt:DT}
This chapter will discuss several variants of the traditional model of decision trees, and minimization algorithms for these models. First, let us recall the original definition of a decision tree.
\begin{defn} \label{defn:DT}
A decision tree is a rooted binary tree in which every non-terminal
node (i.e, not a leaf) is labelled with a variable from $\left\{ x_{i}\right\} _{i=1}^{n}$,
and has out-degree 2. The edges from every such node are labelled
with 0 and 1. Each leaf is labelled with either 0 or 1. A tree $T$
is said to compute the function $f$ if for all $\overline{a}\in\left\{ 0,1\right\} ^{n}$,
the path that begins at the root and follows the edges labelled $a_{i}$
when the node is labelled $x_{i}$, reaches a leaf labelled $f(\overline{a})$. The size of a decision tree is defined as its number of nodes. The depth of a decision tree is the length of the longest path from the root to a leaf.
\end{defn}

Given the truth-table of a function $f$, it is known that one may find a smallest decision tree in polynomial time in the size of the truth-table, denoted $N=2^{n}$ \cite{GLR99}. In the next sections, we show three variants of ordinary decision trees which have efficient truth table minimization algorithms. Symmetric read-once decision trees (SRODT, Section \ref{sctn:SymmetricTrees}) and linear decision lists (LDL, Section \ref{sctn:LDL}) have a polynomial truth table minimization algorithm, while for linear decision trees (LDT, Section \ref{LDT}) we present a quasi-polynomial algorithm. Since some of these algorithms have a similar structure, we present a parametrized generalization of them in Section \ref{sctn:Meta}. Another variant we consider, for which we show that the corresponding decision problem is NP-hard, is that of a decision tree that may contain any function in the nodes (Section \ref{sctn:TreesNPC}).

\section{Efficient Algorithms}

\subsection{Linear Decision Trees} \label{LDT}

We use a technique similar to the one in \cite{GLR99} to devise a quasi-polynomial algorithm for a wider class of trees called linear decision trees (also known as parity decision trees).

\begin{defn}
A linear decision tree (LDT) is a decision tree where every node is labelled by some linear (over $\mathbb{F}_{2}$) function of the input variables. At each node the corresponding linear function is evaluated, and the edge that agrees with its output is followed. An LDT computes a function as explained in Definition \ref{defn:DT}. 
\end{defn}
\indent This model was originally considered by \cite{S81} in the context of integer input, and was later studied also in the boolean case (e.g., \cite{KM93,ZS10}). 

\indent Since the nodes in an LDT are labelled with linear functions, any node of an LDT corresponds to a set of linear constraints, and may be regarded as an affine subspace of $\mathbb{F}_{2}^{n}$. The idea behind the algorithm of \cite{GLR99} for standard decision trees is to find the smallest tree for each cube\footnote{A cube of $\mathbb{F}_{2}^{n}$ that corresponds to $\alpha \in \{ 0,1,\star \}^{n} $ is the set $\{x \in \{ 0,1 \}^{n} \vert \forall i , \alpha_{i}\ne \star \Rightarrow x_{i} = \alpha_{i}\}$.} of the space $\mathbb{F}_{2}^{n}$, starting from cubes that are a single points and up to larger cubes. We do a similar thing with affine subspaces instead of cubes. Our algorithm is based on the following graph, denote by $\mathbb{M}$:
\begin{enumerate}

\item  $\mathbb{M}\ $ has $n+1$ layers, where layer $i$ contains a node for every affine subspace of dimension $n+1-i$. We will identify each subspace $V$ by a pair $[A,b]$ of a matrix $A\in \mathbb{F}_{2}^{n \times n}$ and a vector $b \in \mathbb{F}_{2}^{n}$, such that $V=\{x\vert Ax=b\}$.

\item The edges are only between adjacent layers, and are labelled by some possible test $[u,b]$ (i.e., $\left<u,x\right>=b$ for some $u \in \mathbb{F}_{2}^{n} $ and $b \in \{ 0,1 \}$). The label of an edge $(s,t)$ must be linearly independent in any set of constraints that define the subspace that corresponds to the node $s$.

\item For every edge $(u,v)$ labelled by a constraint $C_{0}$, and for each set of constraints $\{C_{i}\}_{i=1}^{t}$, the subspace $v$ corresponds to the affine subspace defined by $\{C_{i}\}_{i=0}^{t}$.

\end{enumerate}

Notice that the $(n+1)$'th layer of $\mathbb{M}\ $ consists of all 0-dimensional affine subspaces of $\mathbb{F}_{2}^{n}$, namely, all points in $\mathbb{F}_{2}^{n}$, and the first layer consists of one node that corresponds to the entire space. We first show how the construction of $\mathbb{M}\ $ is possible in quasi-polynomial time. To see this, we make two simple observations:

\begin{obsrv} \label{obsrv:CheckEq}
Given two affine subspaces $[A,a]$,$[B,b]$ it is possible to check if they are equal in $poly(n)$ time.
\end{obsrv}

\begin{proof}
It is well known that given $[A,a]$ we may find a basis and a shift vector of the solution space in polynomial time. After doing so to both $[A,a]$,$[B,b]$, we may check equivalence of affine spaces by (say) Gaussian elimination.
\end{proof}

\begin{obsrv} \label{obsrv:NewVectors} Given a set $\{u_{i}\}_{i=1}^{t}$ of independent vectors over $\mathbb{F}_{2}^{n}$ it is possible to produce all vectors that are independent of $\{u_{i}\}_{i=1}^{t}$ in $poly(2^{n})$ time. 
\end{obsrv}

\begin{proof}
We may traverse all possible $2^{n}$ vectors and check if they are linearly dependent in $\{u_{i}\}_{i=1}^{t}$ by traversing all $2^{t} \le 2^{n}$ possible linear combinations of $\{u_{i}\}_{i=1}^{t}$. 
\end{proof}

We denote by CheckEq($[A,a]$,$[B,b]$) the algorithm corresponding to Observation \ref{obsrv:CheckEq} and the one corresponding to Observation \ref{obsrv:NewVectors} by $\mbox{NewVectors}(\{u_{i}\}_{i=1}^{t})$. Using these two algorithms we may construct $\mathbb{M}\ $ (which is of quasi-polynomial size - see below) in quasi-polynomial inductively - begin with constructing the 1st layer, which consists of a single node $v=[0,0]$. In every consecutive step $i$ we traverse every node $v=[A,a]$ in the last constructed layer $i$, apply NewVectors on the set of rows of $A$, and for each output $u$ of NewVectors we create two new nodes in the $(i+1)$'th layer - 
\begin{eqnarray*} 
v_{0} =  \left[ 
					  \begin{pmatrix} A \\ u\end{pmatrix}, 
		           \begin{pmatrix} a \\ 0\end{pmatrix}
			\right]
		,
v_{1} = 	\left[
      			  \begin{pmatrix} A \\ u\end{pmatrix}, 
 					  \begin{pmatrix} a \\ 1\end{pmatrix}
         \right]
\end{eqnarray*}
while preventing duplication with existing nodes by using CheckEq. It is easy to see that there are at most $2^{n+1}$ edges coming out of every node in the graph, and finding them requires $poly(2^n)$ time. However, preventing duplication requires traversing all nodes in the next layer. Therefore, constructing every layer in $\mathbb{M}\ $ can be done in at most $2^n \times |i\mbox{th layer}| \times |(i+1)\mbox{th layer}|$. Therefore the entire algorithm may be done in at most $n \times |\mbox{largest layer}|^2 \times 2^n$. It is widely known that the number of linear subspaces of $\mathbb{F}_{2}^{n}$ of dimension $k$ is given by the Gaussian binomial coefficient (also known as the $q$-binomial coefficient, see definitions in \cite[Chapter 24]{CourseInCombinatorics}), denoted\footnote{The Gaussian binomial coefficient is defined as ${n \brack k} _{2}\triangleq \prod_{i=0}^{k-1}\frac{q^{n-i}-1}{q^{k-i}-1}$ and it is equal to the number of $k$-subspaces of an $n$-dimensional space over a field with $q$ elements.} ${n \brack k} _{2}$. To get the number of affine subspaces of that dimension, we ought to multiply by all possible shift vectors, namely, by $2^n$. The largest binomial coefficient is known to be ${n \choose n/2} _{2}$, which may easily be upper bounded by $2^{O(n^2)}$, which is quasi-polynomial in $N$. Therefore there are at most $2^{O(n^2)}$ vertices in the graph. We now turn to present the algorithm for minimization of linear decision trees. In this algorithm, every node in $\mathbb{M}\ $ will contain the minimal LDT for the corresponding affine subspace. We denote the tree in a node $u$ by $T(u)$. We find the minimal tree for every node inductively, by traversing all possible tests, and checking the resulting trees. The array $M_{v}$ will contain all possible trees for a vertex $v$.

\begin{algorithm}
\DontPrintSemicolon
\LinesNumbered
\nl Construct $\mathbb{M}$ of dimension $n$ ($n$ being the number of variables in $T_{f}$)\;

\nl Label all nodes in the $(n+1)$'th layer with the constant tree according to $f$'s values\;

\nl \For {$i=n, \ldots , 1$}{ \nllabel{MinimizeLDT:LayerLoop}
\nl	\For {all vertices $v=[A,a]$ in layer $i$} {
\nl		\For {every edges $(v,u_{0}),(v,u_{1})$ labelled by $[u,0],[u,1]$ that are connected to $v$} {
\nl			\lIf {$T(u_{0}) = T(u_{1})$}{\nllabel{MinimizeLDT:TreeEq}Add $\min \{ T(u_{0}),T(u_{1}) \}$ to $M_{v}$}\\
\nl			 \Else 
 {\nl Construct a tree with a root labelled $u$ having $T(u_{0}), T(u_{1})$ as sons, and add it to $M_{v}$}
		
		}
\nl		 $T(v) = \min M_{v}$
	}
}
			
\caption{MinimizeLDT($T_{f}$)}\label{MinimizeLDT}
\end{algorithm}

\begin{remark}
The equality between the tree in line \ref{MinimizeLDT:TreeEq} is checked functionally (not topologically) by traversing all assignments. This is the reason that the $\min$ operation is required.
\end{remark}

The correctness of Algorithm \ref{MinimizeLDT} may be verified using the following claim:

\begin{clm}
For every $i \in [n+1]$, after finishing iteration $i$ of the loop in line \ref{MinimizeLDT:LayerLoop}, all nodes of $\mathbb{M}\ $ in layer $i$ contain the minimal LDT for $f$, when restricted to the affine subspace that they represent.
\end{clm}

\begin{proof}
Using induction on $i$, starting from $i=n+1$. The base case is obvious, since $f$ is constant on subspaces of dimension $0$. The induction step may also be seen easily, since all possible tests are taken into consideration.
\end{proof}

The correctness of the algorithm follows from the case $i=1$. The complexity may be seen as quasi-polynomial (in $N$), using a bound on the number of vertices of $\mathbb{M}$. We have that $\mathbb{M}$ consists of $n$ layers, each of size at most $2^{O(n^2)}$ as explained before, thus the total number of vertices in $\mathbb{M}$ is also $2^{O(n^2)}$. Since Algorithm \ref{MinimizeLDT} is polynomial in the number of vertices of the graph, there exists a constant $c>0$ such that the algorithm requires at most $\left(2^{O(n^2)}\right)^{c}=2^{O(n^2)}$ computation steps.

\begin{remark}

We suggest the following restriction of the LDT model, denoted by $LDT_{c}$, which is defined as an ordinary LDT with the additional restriction that every node is only allowed to bare a linear function between up to $c$ variables (for some constant $c$). We devise an efficient algorithm for truth-table minimization of for $LDT_{c}$. We limit the graph $\mathbb{M}$ defined earlier in this section to the graph $\mathbb{M}_{c}$ with the additional restriction that any test $u$ that labels an edge must be of Hamming weight at most $c$. In this case we may asymptotically bound the number of vertices in $\mathbb{M}_{c}$, since the out-degree of any vertex in the graph is polynomial in $n$ and $\mathbb{M}_c$ has $n+1$ layers. The total number of vertices is therefore at most $\sum _{i=0}^{n+1}n^{c\cdot i}=N^{O(\log\log N)}$. Since the construction time of the graph depends on the size of its widest layer, we get a construction algorithm that runs in this time bound.\\ \indent In order to find the minimal $LDT_{c}$ for a given function, we first construct $\mathbb{M}_{c}$ in the same manner as explained earlier in this section, disregarding vectors in the output of NewVectors of Hamming weight over $c$ whenever it is called. Afterwards, we execute Algorithm \ref{MinimizeLDT} on the graph $\mathbb{M}_{c}$. The correctness of the algorithm follows similarly, while the complexity reduces to $N^{O(\log\log N)}$.
\end{remark}

\subsection{Symmetric Read Once Decision Trees} \label{sctn:SymmetricTrees}

Another variant of the ordinary decision tree model that may be constructed in polynomial time under some restrictions is the following:

\begin{defn}
A symmetric decision tree is a decision tree where every node may contain some symmetric function between any number of variables.
\end{defn}

Applying the following restriction on symmetric decision tress allows us to construct a polynomial (in $2^n$) minimization algorithm.

\begin{defn} \label{defn:RRT}
A decision tree (of any kind) will be called a read once tree if any variable appears at most once in any path from the root to a leaf.
\end{defn}
Notice that an ordinary decision tree, as well as read once LDTs (LDTs with the additional restriction described in Definition \ref{defn:RRT}), are a subclass of this class of trees. The algorithm presented here will minimize symmetric read once decision trees (SRODT). \\
\indent The algorithm highly resembles the one of \cite{GLR99}, and we describe it using graph theoretic tools as in the previous section. We denote by $SYM$ the set of all symmetric functions (notice that $|SYM|=2^{n+1}$). For an assignment $\alpha \in \{0,1,\star\}^{n}$, a set $A \subseteq \alpha^{-1}(0) \cup \alpha^{-1}(1)$ and $g\in SYM$ we denote by $g(\alpha\vert_{A})$ the result of applying $g$ on the entries of $\alpha$ that are numbered by elements of $A$. Note that the order of those entries does not matter, since $g$ is symmetric. In order to get a concise description of the algorithm, we define the directed graph $\mathbb{M}^{SYM}$ as follows. The set of nodes corresponds to $\{0,1,\star\}^n$, and $(\alpha,\beta)\in E$ iff 
\begin{eqnarray*}
\exists A \subseteq \alpha ^{-1}(\star), \exists g\in SYM, \exists b\in \{0,1\}\\
\mbox{s.t } \forall i \in \alpha^{-1}(0) \cup \alpha^{-1}(1), \beta_{i} = \alpha_{i}\\
\forall j \in A, \beta_{j} \neq \star \\
\forall k \in \alpha^{-1}(\star) \setminus A, \beta_{j} = \star\\
\mbox{and } g(\beta\vert_{A})=b.
\end{eqnarray*}
Namely, $\beta$ is some extension of $\alpha$ to an assignment that agrees with $\alpha$ in every non-$\star$ entry, has a non-$\star$ entry in every index of some $A \subseteq \alpha^{-1}(\star)$, and such that the function $g$ applied on the entries numbered by $A$ in $\beta$ yields $b$. For convenience, we label each edge with the tuple $(A,g,b)$. Notice that an edge might have more than one label. \\
\indent It is possible to construct the graph $\mathbb{M}^{SYM}$ in $poly(2^{n})$ time. Construct $V=\{0,1,\star\}^{n}$, and traverse all nodes according to decreasing number of stars in the following way - for every node $\alpha$ traverse all nodes $\beta$ that agree with $\alpha$ in every non-$\star$ entry, traverse all $g\in SYM$, apply every $g$ on the appropriate entries of $\beta$ and label the edge accordingly. Notice that the number of vertices is $3^{n}$ and the number of outgoing edges from each node is at most $2^{2n+1}$, thus the size of $\mathbb{M}^{SYM}$ is polynomial in $2^{n}$.\\
\indent Algorithm \ref{MinimizeSRODT} finds the minimal SRODT for a given function $f$, by inductively placing the best SRODT for any cube of $\alpha \in \mathbb{F}_{2}^{n}$ in the node $\alpha$. It begins by assigning the values of the input truth-table to all nodes $\alpha \in \{0,1\}^{n}$. Then it traverses all nodes of $\mathbb{M}^{SYM}$, according to the number of $\star$ entries in them, and checks what is the smallest tree that may be placed in them. The array $W$ is a temporary array used to hold all candidates for the best tree in a node $\alpha$, and will be reset in every iteration of the main loop. For a processed node $\beta$ we denote by $T(\beta)$ the tree that was placed in it. As in the previous section, two trees will be considered equal if they represent the same function. This may be verified in $poly(2^{n})$ time by traversing all assignments.

\begin{algorithm} 
\DontPrintSemicolon
\LinesNumbered
\nl Construct $\mathbb{M}^{SYM}$.\;

\nl $\forall \alpha \in \{0,1\}^{n}$, place the single-leaf tree labelled $(T_{f})_{\alpha}$ in the node $\alpha$.\;
\nl \ForAll {nodes $\alpha$ in $\mathbb{M}^{SYM}$ that haven't been processed yet, and have a minimal number of $\star$ entries} {
\nl 		Reset $W$\;
\nl 		\ForAll{$g\in SYM$, $A \subseteq \alpha^{-1}(\star)$} { \nllabel{MinimizeSRODT:lineLOOP}  
\nl 			\ForAll {$b\in\{0,1\}$} {
\nl 				\ForAll {pairs $\beta,\gamma$ of outgoing neighbours of $\alpha$ such that the connecting edge is labelled by $[g,A,b]$} {
\nl 					\lIf {$T(\beta) \ne T(\gamma))$} {continue to line \ref{MinimizeSRODT:lineLOOP}.} 
					}
\nl				Denote by $T_b$ the smallest tree seen while traversing the pairs $\beta,\gamma$.\;
			}
\nl 			\lIf{$T_{0}=T_{1}$} {$\min (T_{0},T_{1}) \to W$} \\
\nl 			\lElse {Add to $W$ the tree with a root labelled $g$ and $T_{0},T_{1}$ as sons.}
		}
\nl 		$\min W \to \ga$\;
}
\caption{MinimizeSRODT($T_{f}$)} \label{MinimizeSRODT}
\end{algorithm}

The correctness of the algorithm is an easy corollary of the following claim.

\begin{clm}
for every $\alpha \in \{0,1,\star\}^{n}$ with $t$ many $\star$-entries the algorithm puts in $U\left(\alpha\right)$ some minimal SRODT for $f\vert_{\alpha}$. 
\end{clm}

\begin{proof}
By induction on $t$. for $t=0$ it is obvious. For an arbitrary $t$, let $\alpha \in \{0,1,\star\}^{n}$ be some assignment with $t$ stars and let $T$ be some SRODT consistent with $f\vert _{\alpha}$. It is easy to see that the algorithm produces a tree not larger than $T$ after processing node $\alpha$, since it traverses all possible labels for the root of SRODTs for $f \vert _{\alpha}$, one of which is the root of $T$. The rest follows from the induction hypothesis - the key observation which allows us to use the induction hypothesis is that if $T$'s root $v$ is labelled with $g(A)$ then due to the read-once property, the sub-trees rooted at $v$ are some trees that correspond to $f\vert _{\alpha}$ reduced to assignments $\beta$ such that $g(\beta \vert_{A})=0$ (resp. 1) and $(\alpha,\beta)$ is an edge in $\mathbb{M}^{SYM}$, for whom the minimal trees for were already calculated.
\end{proof}

The correctness of the algorithm follows from the case where $t=n$. Furthermore, the algorithm is polynomial since it consists of nesting and concatenations of polynomial loops.

\subsection{Meta-algorithm for Decision Trees with a Fixed Set of Operations in Nodes} \label{sctn:Meta}

The reader may notice that all above algorithms present a similar structure. In this section we try to parametrize the complexity of any minimization algorithm for any fixed set of operations in nodes. \\ 
\indent For a set of operations (or tests) $M=\{M_{i}\}_{i\in [t]}$ (i.e., every $M_{i}$ is a function from some subset of the variables to $\{0,1\}$), we define $\mathcal{M}=\{\mathcal{M}_{\alpha}\}_{\alpha \in \{0,1,\star\}^{t}}$ as the collection of all subsets of $\mathbb{F}_{2}^{n}$ that may be defined by the tests in $M$. Formally, 
\begin{eqnarray*}
\mathcal{M}_{\alpha} = \{ x \in \mathbb{F}_{2}^{n} \vert \forall i \in [t], \alpha_{i} \ne \star \Rightarrow M _{i}(x)=\alpha_{i}\}
\end{eqnarray*}

Notice that, as in the case of LDTs, we might have $\mathcal{M}_{\alpha}=\mathcal{M}_{\beta}$ for $\alpha \ne \beta$. Notice also that if $\{ \{ x\} \vert x\in \mathbb{F}_{2}^{n}\} \nsubseteq \mathcal{M}$, then this model is not universal, since we have two indistinguishable points, thus every function that gives them different values is not computable in this model. \\
\indent In order to generalize the use of the graph $\mathbb{M}$ used in Section \ref{LDT}, we need a generalized notion of independence. 

\begin{defn}
A test $M_{i}$ will be called dependent of a set $\mathcal{M}_{\alpha}$ if either of the following conditions hold:
\begin{enumerate}
	\item $\alpha_{i} \ne \star$.
	\item $\alpha_{i} = \star$ and $\{\mathcal{M}_{\alpha(i=0)}, 		\mathcal{M}_{\alpha(i=1)}\} = \{\emptyset, \mathcal{M}_\alpha\}$ (when $\alpha(i=b)$ denotes the vector $\alpha$ with the $i$th entry changed to $b$).
\end{enumerate}
\end{defn}
Intuitively, a test $M_{i}$ is dependent of a set $\mathcal{M}_{\alpha}$ if the tests that were used to define $\mathcal{M}_{\alpha}$ either contain $M_{i}$ or the value of $M_{i}(x)$ may be derived from them for all $x$.
\\ \indent We now define the directed graph $\mathbb {M}\left( \mathcal{M} \right)$ as follows - 

\begin{enumerate}
\item $V=\mathcal{M}$.
\item An edge $(\mathcal{M}_{\alpha},\mathcal{M}_{\beta})$ exists if there is a test $M_{i}$ independent of $\mathcal{M}_{\alpha}$ and a result $a\in \{0,1\}$ such that $\beta = \alpha_{(i=a)}$ and $\beta _i=a$.

\end{enumerate}

From this stage, generalizing the construction algorithm from \ref{LDT} is straightforward: \\ 
Let $A$ be an algorithm for construction of $\mathbb{M}(\mathcal{M})$. After running $A$ we check if all singleton subsets of $\mathbb{F}_{2}^{n}$ are nodes in the graph. If not, we check if the input function $f$ gives the same value for all indistinguishable points. If not, we reject. If so, we label all singleton sets, as well as sets of indistinguishable points by the single node tree containing $f$'s value. We then apply a bottom up method similar to the one in Section \ref{LDT} - For every node $v$ in $\mathbb{M}(\mathcal{M})$ such that all its sons are already processed, check if all sons represent the same function. If so, copy the smallest tree among them into $v$. If not, choose the outgoing edge labelled with the test that induces the smallest tree, and place it in $v$. After this algorithm finishes, the smallest tree for $f$ will label the node $\mathcal{M}_{\star ^{t}}$.\\
\indent The complexity of the algorithm depends on the structure of $\mathcal{M}$. It is easy to see that the algorithm polynomial in the number of vertices of $\mathcal{M}(\mathbb{M})$. Therefore the total complexity of the minimization algorithm is the complexity of $A$, plus $poly(|\mathcal{M}|)$.

\subsection{Linear Decision Lists} \label{sctn:LDL}
As stated in Theorem \ref{thm:Learning}, learning algorithms may be used for truth-table minimization, provided that a model minimization algorithm that requires $2^{O(n)}$ time exists. In this section we present such model minimization algorithm for linear decision lists, and use a learning algorithm by \cite{BBTV96} to get a polynomial truth table minimization for linear decision lists (Definition \ref{defn:LDL}). 
\begin{remark}
In Theorem \ref{thm:Learning} it is stated that the learning algorithm oughts to be proper and exact. \cite{BBTV96} only mention that their algorithm (Lemma 4.3) is exact. However, the main stage in their algorithm is applying the algorithm of  \cite{HSW90} for learning nested differences of learnable classes. Taking a close look at the algorithm of \cite{HSW90} one may see\footnote{Algorithm ``Total Recall'' in Section 2 of \cite{HSW90}.} that the hypotheses it gives are from the concept class of nested differences, which in our case is a linear decision list.
\end{remark}

\begin{defn} \label{defn:LDL}
A Decision list is a list of pairs $(f_1,v_1),\ldots,(f_r,v_r)$ such that each $f_i$ is a boolean function, each $v_i$ is a value from $\{0,1\}$ and $f_r$ is the constant 1 function. A decision list defines a function $f$ in the following way: for an assignment $a \in \zo{n}$ the value $f(a)$ is equal to $v_i$, where $i$ the least index such that $f_i(a)=1$. 
\end{defn}
Decision lists were first introduced by \cite{Riv87} in the specific case where the $f_i$'s are conjunctions of variables. We will consider a variation of this model which we call \textit{linear decision lists} (introduced by \cite{BBTV96} and denoted there by $\oplus_n$-DL), where each $f_i$ is a linear function (over $\Ftwo$). Notice that linear decision lists may be seen as degenerate linear decision trees with $r-1$ inner nodes and $r$ leaves. The size of a (linear) decision list is defined as its number of inner nodes, excluding the last constant function (e.g., the size of the decision list $\left(f_1,v_1\right),\ldots,\left(f_r,v_r\right)$ is $r-1$). We say that a decision list is \textit{redundant} if it contains a leaf such that no $a \in \zo{n}$ reaches it.

In order to use Theorem \ref{thm:Learning} we need to present a model minimization algorithm. The algorithm we present relies on the following claim:

\begin{clm} \label{clm:RedundantLDL}
Let $S=\left( (f_1,v_1),\ldots,(f_s,v_s) \right), T=\left( (g_1,u_1),\ldots,(g_t,u_t) \right)$ be two non-redundant linear decision lists, both consistent with a function $f$. Then $|S|=|T|$ (i.e., $t=s$).
\end{clm}

\begin{proof}
Assume for contradiction that (w.l.o.g) $t<s$. Observe that $|f^{-1}(1)|$ may be represented in two ways:
\[
|f^{-1}(1)| = \sum _{i=1}^{s-2} v_i \cdot 2^{n-i} + 2^{n-s+1}= \sum _{i=1}^{t-2} u_i \cdot 2^{n-i} + 2^{n-t+1}
\]

To see that, notice that any linear test splits the space into two parts, which are either of equal size, or one of them is the entire space and the other is empty. Since there are no non-reachable leaves, every leaf labelled 1 in depth $i$ contributes exactly $2^{n-i}$ 1's. In addition, exactly one of $v_{s-1},v_s$ and exactly one of $u_{t-1},u_t$ is non- zero. However, this cannot be since $t<s$, and since the binary representation of any number is unique.
\end{proof}

\begin{corollary} \label{cor:LDLModelMinimization}
There exists a polynomial algorithm for the model minimization of linear decision lists.
\end{corollary}

\begin{proof}
Claim \ref{clm:RedundantLDL} allows us to devise the following $2^{O(n)}$ time algorithm which removes any linear dependence between the nodes and removes redundant leaves at the end: Given a linear decision list $L$ with nodes $\{f_i\}_{i=1}^{l}$: for $i=1,\ldots,l$, if $f_i \in span \{f_1,\ldots,f_{i-1}\}$ remove $f_i$ from $L$. If the linear dependence of $f_i$ in $\{f_1,\ldots,f_{i-1}\}$ implies $f_i(a)=1$ for all $a\in \{0,1\}^n$, connect the part of $L$ traversed so far to the $1$ successor of $f_i$, and otherwise to the $0$ successor. At the end of the loop check if both last leaves are identical, if so remove the last test, and check the last leaves again.

This process may easily be seen to conserve the consistency with $f$. Moreover, its output is non-redundant, since all tests are independent and the last two nodes have different values. The output model is minimal according to Claim \ref{clm:RedundantLDL}, since the minimal linear decision list for $f$ is non-redundant as well.
\end{proof}

Therefore, Theorem \ref{thm:Learning} implies:

\begin{corollary}
Linear decision lists have a polynomial truth table minimization algorithm.
\end{corollary}

\begin{remark}
Notice that the corresponding decision problem  
\[
L=\{(T_f,k)\vert f \mbox{ has a linear decision list of size }k\}
\]
is decidable efficiently using an exact learning algorithm which is not proper. In order to do so, we may run the exact learning algorithm, verify that its output indeed represents $f$ by traversing all assignments, and then accept iff $|f^{-1}(1)|$ is divisable by $2^{n-k}$. The correctness of this process is easily provable, since the learning algorithm shows us that there exists \textit{some} linear decision list consistent with $f$, while Claim \ref{clm:RedundantLDL} and the above algorithm show us that any minimal linear decision list for $f$ must be of size $k$, for the minimal $k$ such that $2^{n-k}$ divides $|f^{-1}(1)|$.
\end{remark}

\section{Hardness Results}

\subsection{Decision Trees with Arbitrary Tests in Nodes} \label{sctn:TreesNPC}

In this section the model we consider is that of a tree such that the tests that may be applied over the input in the nodes can be any function. The motivation for this model is exploring the power of the decision tree model, as a function of the tests that are allowed in the nodes. Since this model trivially yields a tree of size 1 to any function, we restrict the minimization algorithm to use only tests that are accepted as input. In our setting the input to the minimization algorithm is already of size $2^{n}$, therefore we may allow the input to contain the specific tests that the algorithm is allowed to put in the nodes, represented as an explicit set in $\{0,1\}^{n}$ (i.e., the test that is represented by a set $D \subseteq \{0,1\}^{n}$ gives 1 to an input $x$ iff $x \in D$). The definition of size for this model will be the number of \textit{different} nodes in the tree (disregarding repetitions of nodes with the same label). We shall see that the corresponding language is NP-hard using a reduction from set cover. A very similar problem was already considered in a different context by \cite{HR76} - the input for the decision tree they define is an element from an abstract finite set $X$, and the tree oughts to supply an exact distinction procedure using subsets of $X$, placed in nodes. They prove that the language of tuples of a set $X$, a set of subsets of $X$ (to be used as tests in nodes), and a number $k$, such that there exists a decision tree of size $k$ that distinguishes between the elements of $X$, is NP-complete.
In our setting the set $X$ may be considered as $\{0,1\}^{n}$ but an exact identification is not needed, since we only need to distinguish between $x$'s in $f^{-1}(1)$ and $f^{-1}(0)$. Moreover, their definition of size is the sum of length of paths in the tree, while we use a completely different notion of size. \\ \indent Define the following language:
\begin{eqnarray*}
L = \{ \left( T_{f}, \{D_{i}\}_{i=1}^{t}, k \right) \vert \mbox{There exists a tree $T$ for $f$ with at most $k$ different tests from $\{D_{i}\}_{i=1}^{t}$ in nodes.}\}
\end{eqnarray*}
Recall the definition of set cover (proved to be NP-complete in \cite{Kar72}):
\begin{eqnarray*}
SC = \{ \left( 1^{m}, \{D_{i}\}_{i=1}^{t}, k \right) \vert \mbox{There are at most $k$ sets from $\{D_{i}\}_{i=1}^{t}$ that cover $[m]$.} \}
\end{eqnarray*}

\begin{remark}
 Notice that the format of the input is crucial to the complexity of deciding the language $SC$. In \cite{Kar72}, the sets $\{D_{i}\}_{i=1}^{t}$ are assumed to contain $\log m$ bit integers in the range $[m]$, and are w.l.o.g assumed to contain all numbers in that range. Therefore the input size is at least $m\log m$, thus we may add $1^{m}$ for convenience without blowing-up the input size.
\end{remark}
\begin{thm}
The language $L$ is NP-hard.
\end{thm}

\begin{proof}
We shall see that $SC \le _{p} L$. The reduction $R$ is as follows - given an instance $\left( 1^{m}, \{D_{i}\}_{i=1}^{t}, k \right)$ of $SC$ we define:
\begin{itemize}
\item $T_{f}=0^{u}1^{m}$ (a truth-table is regarded here as a $2^{n}$ bit string that defines the values of the function on each $x\in \{0,1\}^{n}$ for $n=\log(u+m)$, according to lexicographic order form left to right), when $u>0$ is the smallest positive complement of $m$ to an integer power of 2.
\item If $D_{i}=\{d_{1},\ldots,d_{s}\}$, we define $D_{i}'=\{d_{1}+u,\ldots,d_{s}+u\}$.
\end{itemize}
The output of the reduction is $\left(T_{f},\{D_{i}'\}_{i=1}^{t}, k \right)$.
First, it is easy to see that the reduction $R$ is polynomial: Since $m$ is given in unary, the truth-table $T_{f}$ is at most twice larger than it, and the construction of the sets $D_{i}'$ is obviously polynomial. Second, we show that 
\begin{eqnarray*}
\left( 1^{m}, \{D_{i}\}_{i=1}^{t}, k \right)\in SC \iff R\left( 1^{m}, \{D_{i}\}_{i=1}^{t}, k \right)\in L.
\end{eqnarray*}

For the "if" direction, assume $\left( 1^{m}, \{D_{i}\}_{i=1}^{t}, k \right)\in SC$, and let $\{D_{i}\}_{i\in I}$ be the smallest witness. We define $T$ as a tree with $|I|$ layers, one for each $D_{i}'$. The leaves of the tree are defined as follows: the leftmost leaf (corresponding to the all 0 path) will be labelled 0, the rest will be labelled 1. We now claim that $T$ computes $f$. To see that, we first consider any $x \in \{0,1\}^{n}$ such that $f(x)=1$. From the construction of $R$ it is clear that $x\ge u$ (as numbers in binary representation). Since $\{D_{i}\}_{i \in I}$ is a cover of $[m]$, we have that $\{D_{i}'\}_{i \in I}$ is a cover of $\{u,\ldots,u+m-1\}$ thus there is some $D_{i}'$ such that $D_{i}'(x)=1$. We get that when evaluating $T$ on the input $x$ we follow a 1-edge at some point, and reach a 1-leaf. Second, we consider an $x \in \{0,1\}^{n}$ such that $f(x)=0$. Similarly, $x < u$, therefore $x$ does not belong to any $D_{i}'$. The corresponding computation path in $T$ follows only 0 edges, and reaches a 0 leaf.\\
\indent For the "only if" direction, let $T$ be a tree consistent with $f$ with a node set $\{D_{i}'\}_{i\in I}$. We claim that $\{D_{i}\}_{i\in I}$ is a set cover for $[m]$. Assume for contradiction that $\exists j \in [m] \setminus \cup _{i\in I} D_{i}$, namely, there is a $j$ not covered by $\{D_{i}\}_{i\in I}$. We infer that $j+u$ is not covered by any of the $D_{i}'$s, i.e. $ j+u \in \{u,\ldots,u+m-1\} \setminus \cup _{i\in I} D_{i}'$. Since $j+u$ does not belong to any set in the nodes of $T$, the corresponding computation path follows only 0-edges, and must reach a 1-leaf (since $f(j+u)=1$). However, we have that any $x \in \{0,1\}^{n}$ such that $x<u$ holds $f(x)=0$, and also follows the all 0 path. Therefore $T$ is inconsistent with $f$, a contradiction, and the claim follows.
\end{proof}

\chapter{Boolean and Arithmetic Formulas} \label{chapt:Formulas}
In this chapter we will show both positive and negative results regarding truth-table minimization of boolean and arithmetic formulas. First recall the definition of a boolean formula:

\begin{defn}
A boolean formula is a directed rooted tree, where inner nodes are labelled by $\wedge$ or $\vee$ and leaves are labelled by variables or their negation. Given an assignment in $\{0,1\}^n$, the value of the formula is defined inductively, from the leaves to the root, in the natural way. The size of the formula is its number of nodes, and its depth is the length of the longest path from the root to a leaf.
\end{defn}

\begin{remark}
For some restricted models of formulas we will use a different notion of size (e.g., see Sections \ref{sctn:monDNF}, \ref{subsecn:ROF}).
\end{remark}

When discussing general formulas, we limit the fan-in of the gates to 2, and impose no limitation on depth. When discussing bounded depth formulas, we partition the gates into layers according to depth. In addition, bounding the depth imposes too hard of a restriction unless we allow any fan-in for inner nodes. This allows us to collapse together adjacent inner nodes with the same label, and therefore we also require that all nodes in a layer will bare the same label, alternating between $\vee$ and $\wedge$. We denote by $\Sigma_{k}$ a depth $k$ formula with top gate $\vee$ and by $\Pi_{k}$ a depth $k$ formula with top gate $\wedge$. E.g., a $\Sigma_2$ formula is a DNF. \\
\indent A formula is called \textit {unate} if each variable appears only in its negated form or only in its positive form. A formula is called \textit{monotone} if it is unate, and with no negated variables. 

\begin{defn}
Arithmetic formulas will be defined similarly over some finite field, when $\cdot, +$ gates replace $\wedge,\vee$ gates. A $\Sigma_k^A$ formula is a depth $k$ arithmetic formula with an addition top gate and $\Pi_k^A$ is a depth $k$ arithmetic formula with a multiplication top gate.
\end{defn}

This chapter begins by presenting some harness results and continues with efficient algorithms. Section \ref{sctn:hardness} contains 3 hardness results: In Section \ref{sctn:monDNF} we shall prove that finding the minimal monotone DNF which complies with a given \textit{partial} truth table (i.e., a table with entries from $\{0,1,\star\}$, where $\star$ entries indicate that the function may have either $0$ or $1$ on that assignment) is NP-hard. As will be explained, this problem also arises when discussing monotone $\Pi_3$ formulas. Section \ref{sctn:LowerBoundsSigma3} will show how ``natural'' NP-completeness of the decisional variant of $\Sigma_3$ minimization, if exists, can be used to achieve surprising lower bounds. Section \ref{sctn:FormulaDepth} will use the known $N^{1-\e}$ hardness of approximation of formula size \cite{AKRR03} to achieve a $(1+\frac{1}{c})$ hardness of approximation of formula depth, for some constant $c$. 

In Section \ref{sctn:algorithms}, we give efficient truth-table minimization algorithms for several models, some of which were not previously defined in the literature. The main theorem of this section (Corollary \ref{cor:Main}) concerns the uniqueness of $\wedge,\vee$ and $\oplus$ decompositions of a boolean function into variable disjoint factors. This will be used to construct read-once formulas with gates from $\{\neg,\vee,\wedge,\oplus\}$ in Section \ref{sctn:ROFXOR} (both with and without negligence of the cost of negation gates, see Sections \ref{sctn:CostlyBoolean} and \ref{sctn:CostlyXOR}), boolean and arithmetic read-once formulas (Section \ref{sctn:OurROF}), and finally two models defined by us: unate boolean formulas (Section \ref{sctn:UF2}) and arithmetic formulas (Section \ref{sctn:F2A}) of second order. In these two models any variable participates in a sub-formula of depth 2. The minimization algorithms for these models are based on the decomposition theorem mentioned earlier and depth 2 minimization algorithms, all of them are known except the minimization algorithm for $\Pi_2^A$, which we present in Algorithm \ref{alg:MinimizePi2A}. In all sections, previously known results will be surveyed before presenting our results. We conclude in Section \ref{sctn:OpenProblems} by mentioning some open problems.

\section{Hardness Results} \label{sctn:hardness}

\subsection{Monotone DNF for a partial truth-table} \label{sctn:monDNF}

When size is defined to be the number of terms, it is known that a minimal monotone DNF (that is, a monotone $\Sigma_{2}$ formula) may be easily found using a simple algorithm over the $n$-th dimensional hypercube graph ($G_{HC}=( \{0,1\}^{n} , E)$, where $e=(u,v)\in E$ iff the Hamming distance between $u$ and $v$ is 1, when $u$ is the lighter one). For the ideas behind this algorithm see \cite{A87}. The algorithm is as follows: label each vertex of the graph by the corresponding value of the given truth-table. Find all vertices $\{\alpha_{i}\}_{i=1}^{k}$ that are labelled with 1, and all their incoming neighbours are labelled 0. Output $\alpha= \bigvee_{i=1}^{k}\bigwedge_{j\vert\alpha_{i,j}=1}x_{j}$. One may easily prove that $\alpha$ is indeed the smallest monotone DNF for the given function. Moreover, this algorithm may be used as a black box for unate DNF minimization. To see that, notice first that for any unate function $f$ (that is, a function which has a unate formula representation) there exists $a_{f}\in\{0,1\}^{n}$ such that $f(x\oplus a_{f})$ is monotone (where $\oplus$ denotes bitwise sum modulo 2). Second, notice that a given truth-table may be efficiently verified to represent a monotone function using the above graph $G$ by labelling all vertices accordingly and verifying that there is no directed edge $(u,v)\in E$ such that $f(u)=1$ and $f(v)=0$. Combining these two facts we may construct the minimal unate DNF as follows: given a truth-table of $f$, place its values over the vertices of $G_{HC}$. For every violating edge $(u,v)$ (i.e., $(u,v)\in E$ such that $f(u) = 1$ and $f(v)=0$) such that $u$ and $v$ differ on the $i$th coordinate, define $\left( a_f \right)_i=1$. Later, check if there are violating edges in $G_{HC}$ that corresponds to the function $f(x\oplus a)$. If so reject. Otherwise, apply the ordinary monotone DNF minimization algorithm of the truth-table of $f(x\oplus a_{f})$, and replace in the resulting formula every variable $x_{i}$ such that $(a_{f})_i=1$ with its negation. This algorithm may also be used to find the minimal unate CNF for a function $f$ by applying it over $\overline{f}$, and negating the result. Notice that since the number of edges in the hypercube is $n\cdot 2^{n-1}$, the entire algorithm requires $O(nN)$ time.

As for minimization of ordinary (non-unate) DNF formulas, a classic result \cite{M79} recently simplified in \cite{AH+06} shows:

\begin{thm}
\cite{M79,AH+06} The language 
\begin{eqnarray*}
minDNF = \{(T_{f},k)\vert \mbox{There exists a $k$-term DNF consistent with the truth-table $T_{f}$}\} 
\end{eqnarray*}
is NP-complete.
\end{thm}

The main stage of the reduction in \cite{AH+06} is showing: 

\begin{thm}
\cite{AH+06} The language
\begin{eqnarray*}
minDNF(\star) = \{(T_{f},k)\vert \mbox{There exists a $k$-term DNF consistent with the partial truth-table $T_{f}$\}}
\end{eqnarray*}
is NP-complete, where a partial truth-table is a $2^{n}$-bit string over $\{0,1,\star\}$, and the witness DNF oughts to be consistent with every non-$\star$ entry.
\end{thm}

We show that the monotone variant of $minDNF(\star)$, denoted $minMonDNF(\star)$, is NP-complete using a very similar reduction.

Besides being of independent interest, the problem of deciding $minMonDNF(\star)$ arises in the context of truth-table minimization of monotone $\Pi_{3}$ formulas (notice that a monotone $\Pi_{3}$ formula is a conjunction of monotone DNFs), where the size is the sum of sizes of the 2nd level monotone DNFs. Being unable to show the NP-completeness of the decisional variant of monotone $\Pi_3$ truth-table minimization (denoted $minMon\Pi_3$), we suggest the following relaxation (denoted $minMon\Pi_3^{'}$): The input contains not only a truth-table and a desired size $k$, but also a monotone $\Pi_{3}$ formula $C$. The goal is to decide if there is a $k$-term monotone DNF that may be added to the top-gate of $C$, such that the resulting formula will be consistent with $f$. The reader may easily verify that the problems $minMon\Pi_3^{'}$ and $minMonDNF(\star)$ are equivalent, since the required monotone DNF $M$ that we need to add must satisfy (we denote by $\{M_{i}\}_{i=1}^{k}$ the 2nd level monDNFs of $C$):
\begin{enumerate}
\item $\forall x\in\{0,1\}^{n}$ such that $f(x)=1$, we must have $M(x)=1$.
\item $\forall x\in\{0,1\}^{n}$ such that $f(x)=0$ and $\forall i \in [k],M_{i}(x)=1$, we must have $M(x)=0$.
\item For every other $x\in\{0,1\}^{n}$, the value of $M(x)$ may be arbitrary, thus $M(x)=\star$.

\end{enumerate}
The NP-completeness of $minMonDNF(\star)$ does not imply the NP-completeness of $minMon\Pi_3$, but it could be regarded as an evidence for the possible hardness of the latter. We may also deduce that a minimization algorithm for monotone $\Pi_3$ formulas, if exists, will probably not work in a serial fashion, i.e. it will not construct each of the branches at the 2nd level separately, since in this way the last formula to be constructed might impose an $NP$-complete problem. We leave the NP-completeness of $minMon\Pi_3$ as an interesting open problem (see Section \ref{sctn:LowerBoundsSigma3} for a further discussion about minimization of depth-3 formulas).

The following lemma establishes the NP-completeness of $minMonDNF(\star)$.
\begin{remark} In the following lemma, notice that:
\begin{enumerate}

\item We shall abuse notation by regarding a term in a monotone DNF over the variable set $\{x_{i}\}_{i=1}^{t}$, a $t$-bit binary vector and a subset of $[t]$ as the same object.

\item We treat an arbitrary given family $\bb{S}$ of subsets of $[n]$ as an \textit{anti-chain} (i.e. there are no distinct $S_{i},S_{j}$ such that $S_{i}\subseteq S_{j}$). This does not limit the generality of our claim since our final goal is to find a set cover in $\bb{S}$, thus for any pair $S_i,S_j\in\bb{S}$ such that $S_i \sus S_j$ we may omit $S_i$. Moreover, all containments may be found in $poly(n,|\bb{S}|)$ time.

\item We denote by $G_{HC}$ the hypercube graph, as defined earlier in this section.
\end{enumerate}
\end{remark}

\begin{lemma} \label{lemma:monDNFmainlemma}
Let $\mathbb{S} = \{ S_{i} \} _{i=1}^{s} \subseteq 2^{[n]}$ be an anti-chain such that $\bigcup_{i=1}^s S_i = [n]$. Define sets of vectors in $\{0,1\}^t$ (a vector for each member of $[n]$ and $\bb{S}$, when explicit definition of them and of t later): 

\begin{eqnarray*}
V	=	\{ v^{i} \in \{0,1\}^{t}     \vert i\in [n] \} \\
W	=	\{ w^{S_{i}} \in \{0,1\}^{t} \vert S_{i} \in \mathbb{S}\} 
\end{eqnarray*} 

such that 
\begin{eqnarray} 
\alpha \in S_{i} \Leftrightarrow w^{S_{i}} \le v^{\alpha},\label{eq:condition1} \\
V,W \mbox{ are constant Hamming weight sets.}\label{eq:condition2}
\end{eqnarray}
Let $f:\{0,1\}^{t} \to \{0,1\}$ be the following partial function:

\begin{enumerate}
\item Let $A \triangleq \{ x\in \{0,1\}^{t} \vert \exists i,(x,w^{S_{i}})\in E(G_{HC})\}$. For every $x\in A$ define $f(x)=0$. Namely, $f$ gets 0 on A, which is the set of all vectors identical to some $w^{S_{i}}$ except one missing 1.
\item $\forall v^{i}\in V$ define $f(v^{i})=1$.
\item Let $B \triangleq \{x\in \{0,1\}^t \vert \exists i \in [n], \left( x, v^i \right) \in E\left(G_{HC}\right)\mbox{ and } \nexists j \in [s], w^{S_j}\leadsto x\}$, namely, all vertices that are in-neighbours of a node from $V$, and no node from $W$ leads to them. For all $x\in B$ define $f(x)=0$.
\item Otherwise $f(x)=\star$.
\end{enumerate}

Then $f$ has a monotone DNF with $k$ terms iff there is a set cover of size $k$ in $\mathbb{S}$.
\end{lemma}

\begin{proof} 
For the ``if'' direction, assume $\{S_{i_{j}}\} _{j=1}^{k}$ is a set cover in $\mathbb{S}$. Define the following monotone DNF formula 
\begin{eqnarray*}
\phi=\bigvee_{j=1}^{k}\bigwedge_{\alpha \vert (w^{S_{i_{j}}})_{\alpha}=1}x_{\alpha}. \\
\end{eqnarray*}
To see the consistency of $\phi$ with $f$, let us verify that conditions 1,2,3 above are met. For condition 1, let $x\in A$. We know that there exists an $i$ such that $(x,w^{S_{i}}) \in E(G_{HC})$, thus $x\le w^{S_{i}}$ (bitwise). To see that indeed $\phi\left(x\right)=0$, observe that $\phi$ gets 0 on $x$ iff $x$ does not contain any term of $\phi$. Assume for the contrary that $x$ contains some term of $\phi$, i.e, $\exists j,w^{S_{j}}\le x$. We get that $w^{S_{j}} \le w^{S_{i}}$, a contradiction to (\ref{eq:condition2}).

For condition 2, notice that since $\{ S_{i_{j}}\} _{j=1}^{k}$ is a cover, every $j\in[n]$ has some $S_{i_{j}}$ covering it, and by (\ref{eq:condition1}) $w^{S_{i_{j}}}\le v^{i}$, thus $\phi\left(v^{i}\right)=1$. 

For condition 3, since we only choose terms that correspond to vectors of the form $w^{S_i}$ for some $i$, it is clear that no term covers any $x\in B$, thus $\phi(x)=0$.

For the ``only if'' let
\begin{eqnarray*}
\phi=\bigvee_{i=1}^{k}\bigwedge_{j\in m_{i}}x_{j}.
\end{eqnarray*}
(for some sets $m_{i}\subseteq [t]$) be a minimal monotone DNF consistent with $f$. We claim that for every $m_i$ there exists $j\in [n],k \in [s]$ such that $w^{S_k}\le m_i \le v^j$. First, if $m_i$ is incomparable with all $V$ or strictly larger than any $v^j\in V$ it is redundant, in contradiction with $\phi$'s minimality. Therefore there exists $j\in [n]$ such that $m_i \le v^j$. Second, if $m_i$ is incomparable with all $W$, then since $m_i \le v^j$, we have that $m_i$ covers some node from $B$, which is a 0-node, a contradiction. Moreover, if there is some $k \in [s]$ such that $m_i < w^{S_k}$, then $m_i$ covers a node from $A$, which is also a 0-node, a contradiction.

We construct a set cover in the following manner: since every $m_{i}$ is on some path between a node from $A$ (a 0-node) and $V$ (a 1-node), we may define $\tau(m_{i})$ to be some set $S_{i}$ such that there exists a path between a vertex from $A$ to a vertex from $V$ passing through $m_{i}$ and $w^{S_{i}}$. We claim that $\{ \tau(m_{i})\} _{i=1}^{k}$ is a cover for $[n]$. To see that, let $j\in[n]$. We must show that $\exists i\in[k]$ such that $j \in \tau(m_{i})$. We know that $\phi$ is consistent with $f$, thus $\phi(v^{j})=1$. Therefore $\exists m_{r}$ such that $(\bigwedge_{j\in m_{r} } x_{j})(v^{j})=1$ and $m_{r}\le v^{j}$. According to the definition of $\tau(m_{r})$ and $A$ we have that $\tau (m_{r})\subseteq m_{r}$. Therefore $w^{\tau(m_{r})} \le m_{r}\le v^{j}$, thus $j\in\tau(m_{r})$. 
\end{proof}

We are now ready to show the explicit construction of the sets $V,W$ of the lemma above, and thus, together with some additional technical details, show the NP-completeness of the desired language.

\begin{thm}
The language
\[
minMonDNF(\star) = \{ (T_{f},k) \vert \mbox{There exists a k-term monotone DNF consistent with the partial t.t } T_{f} \}
\] is NP-complete.
\end{thm}

\begin{proof}
We reduce from the language 3PSC (3-Partite set cover) which is the tuples $(n,k,\Pi,\mathbb{S})$ such that $k$ is a natural number, $\Pi$ is a partition of $[n]$ into 3 disjoint sets $\Pi_{1},\Pi_{2},\Pi_{3}$, while $\mathbb{S}=\{S_{i}\}_{i=1}^{s}$
is a collection of subsets of $[n]$ of size exactly 3 (and therefore, also an anti-chain), such that $\forall j\in[3]\forall i \in[s]$ we have $|S_{i}\cap\Pi_{j}|=1$, (namely, every $S_{i}$ has exactly one representative from every $\Pi_{i}$) and there exists a cover of $[n]$ by $k$ elements from $\mathbb{S}$. 3PSC is NP-complete as noted in \cite{AH+06}, by a simple reduction from the 3D matching problem, shown as NP-complete by \cite{GJ78}.\\
\indent Given an instance $(n,k,\Pi,\mathbb{S})$, we produce the vectors $v_{i},i\in[n]$ as follows: let $q$ be the smallest integer such that ${q \choose \frac{q}{2}}\ge n$ (thus $q=\Theta(\log n)$) and $t=3q$. Assign to each $i\in[n]$ some unique $q$-bit vector $b(i)$ that contains exactly $\frac{q}{2}$ 1's. Now, for every $i\in[n]$, let $\Pi(i)$ be the index of the set in $\Pi$ containing $i$. We define $v_{i}\in\{0,1\}^{t}$ by defining it over 3 consecutive $q$-bit blocks. In any block but $\Pi(i)$ it will be all 1's, and it will be $b(i)$ in block $\Pi(i)$. We define the vectors $w^{S_{i}}$ to be the $\wedge$ of all $v^{\alpha}$ such that $\alpha \in S_{i}$. Compute the set $B$ as defined in lemma \ref{lemma:monDNFmainlemma} by preforming BFS from every node in $W$. Define $f$ as in lemma \ref{lemma:monDNFmainlemma} and output $(T_{f},k)$.\\
\indent To see the correctness of the reduction, according to Lemma \ref{lemma:monDNFmainlemma} it suffices to show that the construction meets requirements (\ref{eq:condition1}) and (\ref{eq:condition2}), i.e. that $\alpha \in S_{i}\Leftrightarrow  w^{S_{i}} \le v^{\alpha}$ and $V,W$ are of constant Hamming weight. To see (\ref{eq:condition2}), notice that the Hamming weight of all $w\in W$ is $\frac{3q}{2}$ and the Hamming weight of all $v\in V$ is $2t+\frac{q}{2}$. Second, to prove (\ref{eq:condition1}), assume $\alpha \in S_{i}$. Let $\Pi(\alpha)\in[3]$ be the index of the set of $\Pi$ containing $\alpha$. We have that $w^{S_{i}}$ equals $v^{\alpha}\wedge v^{\beta} \wedge v^{\gamma}$, when $\alpha,\beta,\gamma$ reside in $\Pi_{1},\Pi_{2},\Pi_{3}$ separately. Therefore in block $\Pi(\alpha)$ the entries of $w^{S_{i}}$  are exactly as in $v^{\alpha}$, and in the other blocks the inequality is obvious since $v^{\alpha}$ is 1. Conversely, Assume that $w^{S_{i}} \le v^{\alpha}$. By the construction of $W$, we know that in block $\Pi(\alpha)$, the vector $w^{S_{i}}$ consists of some $q$-bit vector with exactly $\frac{q}{2}$ 1's. In the same block, $v^{\alpha}$ also consists of some $\frac{q}{2}$ 1's $q$-bit vector. Since $w^{S_{i}}\le v^{\alpha}$ implies $(w^{S_{i}})_{j}=1\Rightarrow (v^{\alpha})_{j}=1$, we have that $w^{S_{i}}$ and $v^{\alpha}$ are identical in block $\Pi(\alpha)$. According to the uniqueness of the vectors $b(i)$, we have that $w^{S_{i}}$ was generated by a $\wedge$ that included $v^{\alpha}$, thus $\alpha \in S_{i}$.\\
\indent To see the polynomial complexity, notice that we created a truth-table of size $2^{t}=2^{\Theta(\log n)}=poly(n)$, over $t=\Theta(\log n)$ variables. In order to compute it, we require $n$ computations of vectors in $V$, each takes at most $O(n)$, and $s$ computations of vectors in $W$, each is an $\wedge$ between 3 vectors from V. The computation of the nodes in the set $B$ is done by at most polynomially many runs of BFS.
\end{proof}

\begin{remark}
One may wonder if the above result extends to the arithmetic case. I.e., given a partial truth-table $T_{f}$ and a number $k$, decide if there exists a $k$-term multilinear polynomial over $\Ftwo$ that is consistent with $f$ (as will be noted in further sections, finding a consistent multilinear polynomial for a given full truth-table is possible in polynomial time). We note that this problem may be formulated by purely linear-algebraic means in the following way:
\[
minML(\star)=\{(T_{f},k)\vert T_{f} \in \{0,1,\star\}^{2^{n}}, \exists x\in \{0,1\}^{2^{n}}, \|x\|\le k, B_{T_{f}}Ax=\overline{T_{f}}\}
\]
such that $\overline{T_f}$ is identical to $T_f$ except for $0$s instead of $\star$s, $A$ is the matrix that maps vectors of coefficients of multilinear polynomials to the corresponding truth-tables, $B_{T_f}$ is the diagonal matrix with $1$s in $B_{ii}$ wherever $\left( T_f \right)_i \ne \star$ and $\|\cdot\|$ denotes the Hamming weight. This problem highly resembles certain problems in coding theory (e.g., ``Minimum Distance'' in \cite{V97}), and exploring this similarity may be an interesting research problem. 
\end{remark}

\subsection{Lower Bounds from Natural NP-completeness of $min\Sigma_{3}$} \label{sctn:LowerBoundsSigma3}

The Minimum Circuit Size Problem (MCSP) is the language of pairs $(T_{f},k)$ such that $T_{f}$ is a truth-table, and there exists a boolean circuit of at most $k$ gates, with fan-in limited to 2, that is consistent with $f$. In \cite{KC00} it is proved that NP-completeness (which is ``natural'' in some sense, see definition \ref{defn:natural}) of this language would yield explicit constructions of functions with high circuit complexity, under some reasonable assumptions. Since finding such constructions is a long standing open problem, we may deduce that such a reduction may be hard to find. See discussion in \cite{AB09}, Section 14.4.3.\\
\indent In this section we show that even showing a ``natural'' NP-completeness of the truth-table minimization problem of a much simpler model would still yield an explicit construction of functions with surprisingly high circuit complexity. Specifically, let $min\Sigma_{3}$ be the language of pairs $(T_{f},k)$ such that $T_{f}$ is a truth-table, and there exists a $\Sigma_{3}$ formula with at most $k$ gates that is consistent with $f$. We show that if there is a ''natural'' reduction from any NP-complete language to $min\Sigma_{3}$, then it is possible to explicitly construct a boolean function on $m$ inputs that has no linear-size logarithmic-depth circuits, under the assumption that $\mbox{NP}\nsubseteq\mbox{SUBEXP}$ (when $\mbox{SUBEXP}\triangleq \bigcap_{\e > 0} \mbox{DTIME}(2^{n^{\e}})$). First, let us define the kind of reductions that we consider.
\begin{defn} \label{defn:natural}
\cite{KC00} A polynomial reduction $R$ from a language $A$ to a language $B$ is called \textit{natural} if for every instance $I$ of $A$, the size of $R(I)$, as well as any numerical parameters of it, is a function of $|I|$ only. 
\end{defn}

As noted in \cite{KC00}, most known reductions are natural. In order to prove the main result of this section, we will need the following lemmas:

\begin{lemma} \label{lemma:NumberSigma3}
Denote by $\Sigma_{3}(s,t)$ the number of $\Sigma_{3}$ formulas with $s$ gates on $t$ variables. Then $\Sigma_{3}(s,t)\le(2^{2t})^{s}\cdot s^{s+1}$.
\end{lemma}
 
\begin{proof} First, we have to choose what will be the size of the bottom level (the sizes of the middle and top levels are determined by it). Clearly, there are $s$ possibilities for that. Second, consider the gates at the bottom level. For every gate we ought to choose some subset of $\{ x_{i}\} _{i=1}^{t}\cup \{ \overline{x_{i}}\} _{i=1}^{t}$ as inputs. There are $2^{2t}$ options for every gate, therefore at most $(2^{2t})^{s}$ options overall. The last stage will be connecting every gate in the bottom level to some gate at the middle level. There are at most $s$ gates, for each of them we have at most $s$ options. Overall we get that $\Sigma_{3}(s,t)\le (2^{2t})^{s}\cdot s^{s+1}$.
\end{proof}

We now cite (without proof) a celebrated result by Valiant \cite{V77}, recently simplified by \cite{V09}.

\begin{lemma} \label{lemma:Jukna}
If $f:\{0,1\}^{m}\to\{0,1\}$ cannot be computed by $\Sigma_{3}$ formulas of size $2^{O(m/\log\log m)}$ then f cannot be computed by boolean circuits (of fan-in 2) of depth $O(\log m)$ and size $O(m)$.
\end{lemma}

These lemmas give us the following corollary:

\begin{thm} \label{thm:LowerBounds}
If $\mbox{NP}\nsubseteq\mbox{SUBEXP}$ and there exists a natural reduction $\mbox{SAT}\le_{p}min\Sigma_{3}$, then there is an explicitly defined family of functions $\mathcal{F} = \{ f_{n}\} _{n\in\mathbb{N}}$ that cannot be computed by circuits of depth $O(\log n)$ and size $O(n)$ infinitely often.
\end{thm}

\begin{proof}
Let $R$ be the natural reduction from $\mbox{SAT}$ to $min\Sigma_{3}$. Denote $\phi\overset{_{R}}{\mapsto}(T_{\phi},s(|\phi|))$ (the existence of the function $s$ is guaranteed from $R$ being natural). Denote $|\phi|=n$, and notice that $T_{\phi}$ is a truth-table of a function on $c\log n$ variables for some $c>0$, since $R$ is a polynomial reduction. \\
\indent Now, if for all $\e>0$ we have $s(n)=o(n^{\e})$, we show that $\mbox{NP}\subseteq\mbox{SUBEXP}$. According to lemma \ref{lemma:NumberSigma3} we get that $\Sigma_{3}(s(n),c\log n)=o(2^{n^{2\e}})$ (see Lemma \ref{lemma:LowerBoundsAddendum} at the end of this chapter) for every $\e>0$. Therefore we may traverse all relevant $\Sigma_{3}$ formulas deterministically using an algorithm from the class $\mbox{SUBEXP}$. We will use that fact to decide $\mbox{SAT}$ in \mbox{SUBEXP} by applying the reduction $R$ on any instance $\phi$, and decide if $(T_{\phi},s(|\phi|))$ is in $min\Sigma_{3}$ by traversing all relevant formulas and checking in polynomial time if any of them is consistent with $T_{\phi}$ (this is possible since $|T_{\phi}|=poly(n)$). \\
\indent We infer that there exists $\e>0$ such that $s(n)=\Omega(n^{\e})$ for every $n\in N$ for some infinite $N\subseteq\mathbb{N}$. We construct the desired function family in the following manner: take any family of non-satisfiable $\mbox{CNF}$ formulas $F=\{ \phi_{n}\} _{n\in\mathbb{N}}$, apply the reduction $R$ on it and define $\mathcal{F}$ to be the resulting family of functions. Since $F$ is a family of no instances of $\mbox{SAT}$, $\{R(\phi_{n})\} _{n\in\mathbb{N}}$ is a family of no instances of $min\Sigma_{3}$. We deduce that infinitely often we have $f_{n}\in{\mathcal{F}}$  such that its $\Sigma_{3}$ complexity is at least $2^{\e \log n}$ (notice that $f_{n}$ is a function over $c\log n$ variables for some $c>0$). Therefore, by lemma \ref{lemma:Jukna} we get the desired family of functions over $m=\log n$ inputs that cannot be computed by circuits of depth $O(\log m)$ and size $O(m)$.
\end{proof}

\subsection{Inapproximability of Formula Depth} \label{sctn:FormulaDepth}

Alledner et al. show an inapproximability result regarding formula size: (Theorem 25 in \cite{AKRR03}) the minimal boolean formula size (of fan in 2) for a given truth-table is inapproximable in $\mbox{BPP}$ up to a factor of $N^{1-\e}$, when $N$ is the size of the input, and for every $0<\e<1$. A hardness result on depth minimization of formulas when the fan-in is bounded by two\footnote{Bounding the fan-in is necessary since without a bound the optimal formula depth for any function is 2, e.g. by a canonical CNF.} is easily derivable from \cite{AKRR03} using formula balancing. According to \cite{K78}, given a formula of size $l$ for a function $f$, we may balance it to get an equivalent formula of depth $c\cdot \log(l)$ for $c=1.73$. Therefore, if we denote by $\mbox{FSIZE}(f)$ and $\mbox{FDEPTH}(f)$ the optimal formula size and depth of $f$, we have $\mbox{FDEPTH}(f)\le c\log (\mbox{FSIZE}(f))$. Moreover, since the fan in is at most 2, we have $\mbox{FSIZE}(f)\le 2^{\mbox{FDEPTH}(f)}$. 

\begin{thm}
Assuming that Blum Integer Factorization is not in $\mbox{ZPP}$, there is no polynomial algorithm that approximates $\mbox{FDEPTH}(f)$ up to a factor of $1  +\frac{1}{c}$.
\end{thm}

\begin{proof}
Let $D$ be a polynomial time algorithm such that on input which is a truth-table $T_{f}$ of a function $f$ such that $\mbox{FDEPTH}(f)\le (1-\e)n$ for some $\e>0$ the algorithms outputs 
\begin{eqnarray*}
\mbox{FDEPTH}(f)\le D(T_{f})\le(1+\frac{1}{c})\cdot \mbox{FDEPTH}(f)
\end{eqnarray*}
Notice that the existence of $D$ is a weaker assumption than the existence of an approximation algorithm for any $f$. We get that 
\begin{eqnarray*}
\mbox{FDEPTH}(f)\le D(T_{f})\le (1-\e)\cdot n + \frac{1}{c}\cdot \mbox{FDEPTH}(f)
\end{eqnarray*}
Thus 
\begin{eqnarray*}
2^{FDEPTH(f)}\le 2^{D(T_{f})}\le 2^{(1-\e)\cdot n + \frac{1}{c}\cdot FDEPTH(f)} \\
\end{eqnarray*}
And using the aforementioned bounds will yield
\begin{eqnarray*}
\mbox{FSIZE}(f)\le 2^{D(T_{f})} \le N^{1-\e}\cdot \mbox{FSIZE}(f)
\end{eqnarray*}
Therefore defining an algorithm $A$ such that $A(T_{f})=2^{D(T_{f})}$ will contradict the result from \cite{AKRR03} mentioned earlier. 
\end{proof}

\section{Efficient Algorithms} \label{sctn:algorithms}

\subsection{Background on Read Once Formulas and Related Models} \label{subsecn:ROF}
A \textit{Boolean} read-once formula (abbr. $ROF$) is a boolean formula in which every variable labels at most one leaf (in its negated or non-negated form).  If the set of operations in nodes is $\{\cdot,+\}$ over some fixed finite field it is called an \textit{Arithmetic} $ROF$. Both arithmetic and boolean formulas were a subject to extensive research in many fields of complexity theory (learning theory in \cite{AHK89,BHH92} and polynomial identity testing in \cite{SV08} to name a few).  \\
\indent It is known (as mentioned in \cite{AHK89}) that the boolean and arithmetic ROF for a specific function is unique\footnote{In the arithmetic case over $\bb{F}_2$, it is unique only up to negated variables. Since negation is commonly referred to as a cheap operation, the cost of the negation gates is usually neglected. See further section for discussion about minimization with costly negation gates}, and therefore the algorithms we look for in this case are rather \textit{decision} algorithms than \textit{minimization} algorithms. In this section we provide an algebraic proof for a stronger claim, that will allow us to devise an efficient truth-table minimization algorithms for several classes of formulas that include ROFs over several different bases. Moreover, we discuss the natural case where the variables appear only in their non-negated form, and negation gates are costly.

\subsubsection{Previous Work}
As mentioned earlier, it is known that boolean $ROF$s \cite{AHK89} and arithmetic ROFs over any field \cite{BHH92} may be learned efficiently using membership and equivalence queries. We shall utilize these facts to construct efficient truth-table minimization algorithms. However, Golumbic et al. \cite{GAU04} show a stronger result than ours: a boolean read-once function (i.e., a function that has a boolean $ROF$) may be recognized in time $O(n\cdot k)$, where $k$ is the number of terms in some DNF representation of $f$. Since every read-once function is also unate (i.e., it has a corresponding unate formula), we may construct its unate DNF representation (as mentioned in Section \ref{sctn:monDNF}) and feed it into the algorithm of \cite{GAU04}. The advantage of our algorithm present over the existing ones (including the algorithm of \cite{P95}) is that it constructs a $ROF$ over larger bases, (e.g., ${\vee, \wedge, \oplus}$) and is extendible to wider classes of formulas (see Definition \ref{defn:UFK}).

\subsection{Main Theorem}

We now turn to formulate the main theorem of this section. This theorem will allow us to devise several truth-table minimization algorithms. Moreover, it also proves the uniqueness of ROFs over any basis $B\subseteq \{\wedge,\vee,\oplus\}$ over $\Ftwo$. First, a few definitions:

\begin{defn}
Let $\odot$ be some symmetric operation over $\Fq$ on arbitrarily many inputs. $f\left(x_{1},\ldots,x_{n}\right)$ is called $\odot$-decomposable if there is a non-trivial partition of $\left\{ x_{i}\right\} _{i=1}^{n}=\bigcup_{i=1}^{t}X_{i}$ (i.e., $t>1$, $\forall i,X_{i}\ne\emptyset$ and $\forall i\ne j,X_{i}\cap X_{j}=\emptyset$) such that $f=\bigodot_{i=1}^{t}f_{i}\left(X_{i}\right)$ for some functions $f_{i}$. A $\odot$-decomposition $f=\bigodot_{i=1}^{t}f_{i}\left(X_{i}\right)$ is maximal if $\forall i,f_{i}$ is $\odot$-indecomposable. The functions $\left\{ f_{i}\right\} _{i=1}^{t}$ are called the factors of $f$. In out settings $\odot$ will either be the boolean $\wedge,\vee$ or the operations $+,\cdot$ over a finite field $\Fq$.
\end{defn}

Now, we use algebraic tools to show the following two theorems:

\begin{thm} \label{thm:mainWedgeVee}
Let $f:\mathbb{F}_{2}^{n}\to\mathbb{F}_{2}$. Then $f$ cannot be both $\wedge$ and $\vee$-decomposable.
\end{thm}

\begin{thm} \label{thm:mainWedgeXor}
Let $f:\mathbb{F}_{2}^{n}\to\mathbb{F}_{2}$. Then $f$ cannot be both $\wedge$ and $\oplus$-decomposable.
\end{thm} 

As a corollary we get:

\begin{corollary} \label{cor:Main}
Let $f:\mathbb{F}_{2}^{n}\to\mathbb{F}_{2}$. At most one of the following is true: 
\begin{enumerate}
\item $f$ is $\wedge$-decomposable.
\item $f$ is $\oplus$-decomposable.
\item $f$ is $\vee$-decomposable.

\end{enumerate}
\end{corollary}

In order to prove Theorems \ref{thm:mainWedgeVee} and \ref{thm:mainWedgeXor}, we present the following definitions and observations (we denote by $f^{ML}$ the unique multilinear polynomial consistent with $f:\mathbb{F}_{2}^{n}\to\mathbb{F}_{2}$):

\begin{obsrv} \label{obsrv:fMLirreducible}
Let $f:\mathbb{F}_{2}^{n}\to\mathbb{F}_{2}$. $f$ is $\wedge$-decomposable iff $f^{ML}$ is reducible.
\end{obsrv}
\begin{proof}
Assume $f$ is $\wedge$-decomposable. Let $f=\bigwedge_{i=1}^{k}f_{i}\left(X_{i}\right)$ be the maximal decomposition. Represent every $f_{i}$ as a multilinear polynomial, denoted by $f_{i}^{ML}$. Obviously, $\prod_{i=1}^{k}f_{i}^{ML}$ is a multilinear polynomial (since the $f_{i}$'s are variable disjoint) consistent with $f$, and by the uniqueness of the multilinear representation it is exactly $f^{ML}$. In addition, every $f_{i}^{ML}$ is irreducible, since otherwise $f_{i}$ would be $\wedge$-decomposable, contradicting the definition of a maximal decomposition. \\ \indent Conversely, assume $f^{ML}$ is reducible. Let $f^{ML}=\prod_{i=1}^{k}f_{i}$ be the factorization. Since $f^{ML}$ is multilinear, then the $f_{i}$'s are variable disjoint, and therefore $f$ is $\wedge$-decomposable.
\end{proof}

As a simple corollary, we have:
\begin{corollary} \label{cor:DecomposableReducible}
Let $f:\mathbb{F}_{2}^{n}\to\mathbb{F}_{2}$. $f$ is $\vee$-decomposable iff $f^{ML}+1$ is reducible.
\end{corollary}

Moreover, since the ring of polynomials over a field is a unique factorization domain, we have:
\begin{corollary} \label{cor:UniquenessOfDecomposition}
Let $f:\mathbb{F}_{2}^{n}\to\mathbb{F}_{2}$. Then its maximal $\wedge$ (resp. $\vee$) decomposition is unique (if exists).
\end{corollary}

Now, we give the algebraic tools to be used in the proof of Theorems \ref{thm:mainWedgeVee} and \ref{thm:mainWedgeXor}.

\begin{defn} \cite[Definition 4.5]{SV10} For any polynomial $P\in\mathbb{F}\left[x_{1},\ldots,x_{n}\right]$, define the multilinear commutator of $x_{i}$ and $x_{j}$ as: 
\begin{eqnarray*}
\Delta_{ij}P=P\vert_{\begin{array}{c}
x_{i}=1\\
x_{j}=1
\end{array}}\cdot P\vert_{\begin{array}{c}
x_{i}=0\\
x_{j}=0
\end{array}}-P\vert_{\begin{array}{c}
x_{i}=1\\
x_{j}=0
\end{array}}\cdot P\vert_{\begin{array}{c}
x_{i}=0\\
x_{j}=1
\end{array}}
\end{eqnarray*}
\end{defn}
 
\begin{defn} \cite[Definition 4.2]{SV10} A polynomial $Q$ will be called $\left(x_{i},x_{j}\right)$-decomposable if $Q$ does not have any irreducible factor depending both on $x_{i}$ and $x_{j}$.
\end{defn}

\begin{defn} \label{defn:partialDerivative} \cite[Definition 2.5]{SV08} For a polynomial $Q\in\mathbb{F}\left[x_{1},\ldots,x_{n}\right]$ define the discrete partial derivative with respect to $x_{i}$ as: $\frac{\partial Q}{\partial x_{i}}=Q\vert_{x_{i}=1}-Q\vert_{x_{i}=0}$.
\end{defn}

We will use the following propositions:

\begin{prop} \label{prop:DeltaQis0}
\cite[Lemma 4.6]{SV10} Let $Q$ be a multilinear polynomial. Then $Q$ is $\left(x_{i},x_{j}\right)$-decomposable $\iff\Delta_{ij}Q=0$.
\end{prop}

\begin{prop} \label{prop:QplusCisSmthn}
\cite[Observation 3.2.14]{V12} If $Q$ is a multilinear polynomial over a field $\mathbb{F}$, then $\forall c\in\mathbb{F},\Delta_{ij}\left(Q+c\right)=\Delta_{ij}Q+c\cdot\frac{\partial^{2}Q}{\partial x_{i}\partial x_{j}}$. \\ \indent Note: $\frac{\partial^{2}Q}{\partial x_{i}\partial x_{j}}=\frac{\partial^{2}Q}{\partial x_{j}\partial x_{i}}$ by lemma 2.6 in \cite{SV08}. 
\end{prop}

Finally, before proving Theorem \ref{thm:mainWedgeVee}, we note the following two simple lemmas, whose proof appears at the end of this chapter:

\begin{prop} \label{prop:2Partitions}
For any two non-trivial partitions $\left\{ X_{i}\right\} _{i=1}^{k},\left\{ X_{i}'\right\} _{i=1}^{k'}$ of $\left\{ x_{i}\right\} _{i=1}^{n}$, there exists $i,j\in\left[n\right]
 , i_{1},j_{1}\in\left[k\right]$ and $i_{2},j_{2}\in\left[k'\right]$ such that $i\ne j,i_{1}\ne j_{1},i_{2}\ne j_{2}$ and $x_{i}\in X_{i_{1}}\cap X_{i_{2}}',x_{j}\in X_{j_{1}}\cap X_{j_{2}}$. 
 \end{prop}

Intuitively, Proposition \ref{prop:2Partitions} implies that we may choose a pair of variables $x_{i}\ne x_{j}$ such that each of them is in the intersection of two different sets from the above partitions.

\begin{prop} \label{prop:Delta2Not0}
If $Q\in\mathbb{F}\left[x_{1},\ldots,x_{n}\right]$ is $\left(x_{i},x_{j}\right)$-decomposable and multilinear, then $\frac{\partial^{2}Q}{\partial x_{i}\partial x_{j}}\ne 0$. Note that since $Q$ is multilinear, so are its partial derivatives, therefore $\frac{\partial^{2}Q}{\partial x_{i}\partial x_{j}}\ne 0$ both as polynomials and as functions.
\end{prop}

Now, using the above we obtain the following more general lemma. Theorem \ref{thm:mainWedgeVee} will be an easy corollary of it.

\begin{lemma} \label{lemma:PplusCirr}
Let $P\left(x_{1},\ldots,x_{n}\right)$ be any multilinear polynomial over any field $\mathbb{F}$. Assume that the following non-trivial and maximal factorizations exist, for some $c\in\mathbb{F}$: 
\begin{eqnarray*}
P	=	\prod_{i=1}^{k}P_{i}\left(X_{i}\right) \\
P+c	=	\prod_{i=1}^{k'}P_{i}'\left(X_{i}'\right)
\end{eqnarray*}
Then $c=0$.
\end{lemma}

\begin{proof}
According to Proposition \ref{prop:2Partitions}, we may consider the partitions $\left\{ X_{i}\right\} _{i=1}^{k}$ and $\left\{ X_{i}'\right\} _{i=1}^{k'}$ and find $x_{i},x_{j}$
such that both $P$ and $P+c$ are $\left(x_{i},x_{j}\right)$-decomposable. From Proposition \ref{prop:DeltaQis0}, we get that $\Delta_{ij}\left(P\right)=\Delta_{ij}\left(P+c\right)=0$, and thus, by Proposition \ref{prop:QplusCisSmthn} we get $\Delta_{ij}\left(P+c\right)=\Delta_{ij}P+c\cdot\frac{\partial^{2}P}{\partial x_{i}\partial x_{j}}=c\cdot\frac{\partial^{2}P}{\partial x_{i}\partial x_{j}}=0$. Now, from Proposition \ref{prop:Delta2Not0} and since $P$ is $\left(x_{i},x_{j}\right)$-decomposable, we get that $\frac{\partial^{2}P}{\partial x_{i}\partial x_{j}}\ne 0$, and therefore $c=0$. 
\end{proof}

\begin{proof} (of Theorem \ref{thm:mainWedgeVee}) Assume for contradiction that $f$ is both $\wedge$ and $\vee$-decomposable. Using negation we infer that both $f$ and $\overline{f}=f\oplus 1$ are $\wedge$-decomposable. By Corollary \ref{cor:DecomposableReducible} we have that both $f$ and $f\oplus 1$ are decomposable, which contradicts lemma \ref{lemma:PplusCirr}.
\end{proof}

\begin{proof} (of Theorem \ref{thm:mainWedgeXor}) Assume for contradiction that $f$ is both $+$-decomposable and $\cdot$-decomposable (we use the algebraic notation $\cdot,+$ rather than the boolean one $\wedge,\oplus$ in this proof for convenience). Therefore we may write $f_{1}\left(X_{1}\right)+f_{2}\left(X_{2}\right)=g_{1}\left(Y_{1}\right)\cdot g_{2}\left(Y_{2}\right)$ such that $X_{1}\cap X_{2}=Y_{1}\cap Y_{2}=\es,X_{1}\cup X_{2}=Y_{1}\cup Y_{2}=\left\{ x_{i}\right\} _{i=1}^{n}$. We shall use the fact that Definition \ref{defn:partialDerivative} complies with the product and sum rules for multilinear polynomials, as ordinary derivative does \cite[Lemma 2.1.9]{V12}. By Proposition \ref{prop:2Partitions} we may choose $x_{i}\in X_{1}\cap Y_{1},x_{j}\in X_{2}\cap Y_{2}$. Since any multilinear polynomial $P$ depends on a variable $x_{k}$ iff $\frac{\partial P}{\partial x_{k}}\ne 0$ \cite[Lemma 2.1.8]{V12} we get 
\[
\frac{\partial^{2}f^{ML}}{\partial x_{i}\partial x_{j}}=\frac{\partial^{2}f_{1}}{\partial x_{i}\partial x_{j}}+\frac{\partial^{2}f_{2}}{\partial x_{i}\partial x_{j}}=0
\]
since $f_{1}$ does not depend on $x_{j}$, and $f_{2}$ does not depend on $x_{i}$. However, since $f^{ML}=g_1(Y_1)\cdot g_2(Y_2)$, we know that $f^{ML}$ is $(x_i,x_j)$-decomposable, and by Proposition \ref{prop:Delta2Not0} we get a contradiction.
\end{proof}

We shall now deduce Corollary \ref{cor:Main} using Theorems \ref{thm:mainWedgeVee} and \ref{thm:mainWedgeXor}.

\begin{proof} (of Corollary \ref{cor:Main}) We prove the following three claims for $i\in\left[3\right]$: If claim $i$ holds then claims $\left[3\right]\backslash\left\{ i\right\} $ does not hold.
\begin{itemize}
\item $i=1$. If 1 holds then $f$ is $\wedge$-decomposable, thus it cannot be $\vee$-decomposable due to Theorem \ref{thm:mainWedgeVee}, and 3 does not hold. Moreover, if 2 does hold we have that $f$ is both $\cdot$ and $+$ decomposable in $\bb{F}_{2}$, contradicting Theorem \ref{thm:mainWedgeXor}.
\item  $i=2$. If 2 holds then similarly, 1 does not hold. If 3 holds, then $\overline{f}$ is $\wedge$-decomposable, and we have: $\bigwedge_{i=1}^{t}\overline{g_{i}}\left(Y_{i}\right)=\bigoplus_{i=1}^{k}f_{i}\left(X_{i}\right)+1$ Thus $\overline{f}$ is both $+$-decomposable and $\cdot$-decomposable over $\bb {F}_{2}$, contradicting Theorem \ref{thm:mainWedgeXor}.
\item $i=3$. 1 cannot hold due to Theorem \ref{thm:mainWedgeVee}. 2 cannot hold from the same reason as in the previous case.

\end{itemize}
\end{proof}

Before presenting the algorithm for minimization of read-once formulas over $\{\vee,\wedge,\oplus\}$, we give the following useful observations:

\begin{obsrv} \label{obsrv:HowToFindDecompositions} \ 
\begin{enumerate}
\item $\wedge$-decomposability may be found in $poly(2^n)$ time. By defining 
\begin{eqnarray*}
\sum _{\ga \in\zo{n}} f(\ga)\cdot \prod_{i\vert\ga_i = 1} x_i \prod_{i\vert\ga_i =0}(x_i\oplus 1)
\end{eqnarray*}
and changing basis, we find the unique multilinear representation of $f$. Observation \ref{obsrv:fMLirreducible} shows that we may factor this polynomial (e.g., in $O(N^8)$ by the algorithm of \cite{Len85}) to find the multilinear representations of $f$'s factors, and convert them back to their truth-table representation. Notice that if $f$ is unate, a simpler algorithm is possible: find $f$'s minimal unate DNF. Construct a graph $G$ on $n$ vertices such that $\{i,j\}\in E(G) \Leftrightarrow $ there exists a term containing both $x_i,x_j$. Finding connected components in this graph will suffice.

\item Similarly, $\vee$-decomposability may be found in $poly(2^n)$ time.

\item $\oplus$-decomposability may also be found in $poly(2^n)$ time using an algorithm that searches connected components in the graph $G'$, defined according to the multilinear representation of $f$, as $G$ was constructed according to the unate DNF in Section 1.
\end{enumerate}
\end{obsrv}

\begin{remark}
Unlike $\wedge$ and $\vee$-decompositions, $\oplus$-decomposition is not unique, since constants may be distributed among the factors. Therefore we define a \textit{canonical} $\oplus$-decomposition, to be the decomposition that corresponds to $f^{ML} = \bigoplus_{i=1}^k f^{ML}_i \oplus c$ such that all of $f_i^{ML}$ are homogeneous polynomials.
\end{remark}

\subsection{Read Once Formulas Over $\{\neg,\oplus, \vee, \wedge\}$} \label{sctn:ROFXOR}

We now turn to present the efficient minimization algorithm for read-once formulas over $\Ftwo$ with the operation set $\{\vee, \wedge ,\oplus\}$.

\begin{defn} \label{defn:ROFoPLUS}
Denote by $ROF_{\oplus}$ the class of read once formulas with gates labelled $\vee,\wedge$ or $\oplus$, where the variables appear in their negated or positive form and no constants are allowed. The size of an $ROF_{\oplus}$ formula is defined to be its number of gates.
\end{defn}

The following algorithm finds the minimal $ROF_\oplus$ representation for a function $f$ given as its full truth-table. The algorithm works in a recursive manner, trying to find $\vee,\wedge$ or $\oplus$-decomposition as in Observation \ref{obsrv:HowToFindDecompositions}. In the case where $f$ is $\wedge$ or $\vee$-decomposable, the factorization is unique, and the algorithm proceeds recursively on $f$'s factors. If $f$ is $\oplus$-decomposable, it will find the canonical decomposition $f^{ML} = \bigoplus_{i=1}^k f^{ML}_i \oplus c$, and proceed in a recursive manner on the factors $f_1+c$ and $f_i$ for $i=2,\ldots,k$. As we shall see, negation does not change the size of the formula, and therefore this will suffice for our case.

\begin{algorithm}
\DontPrintSemicolon
Construct $f^{ML}$, the multilinear representation of $f$.\;

\lIf {$f$ is over a single variable $x_i$} {return either $x_i,\overline{x_i}$.}\\

\lIf {$f$ is $\wedge$-decomposable as $f=\bigwedge_{i=1}^k f_i$} {\Return $\bigwedge _{i=1}^{k} MinimizeROF_{\oplus}(T_{f_i})$.}\\

\lIf {$f$ is $\vee$-decomposable as $f=\bigvee_{i=1}^k f_i$} {\Return $\bigvee _{i=1}^{k} MinimizeROF_{\oplus}(T_{f_i})$.}

\If {$f$ is $\oplus$-decomposable as $f=[\bigoplus_{i=1}^k f_i]\oplus c$} {\Return $[\bigoplus _{i=2}^{k} MinimizeROF_{\oplus}(T_{f_i})] \oplus MinimizeROF_{\oplus}(T_{f_1\oplus c})$.}

Reject.\;

\caption{$MinimizeROF_{\oplus}(T_{f})$} \label{alg:MinimizeROFoPLUS}
\end{algorithm}

\begin{clm}
Algorithm \ref{alg:MinimizeROFoPLUS} finds the minimal $ROF_{\oplus}$ representation of the given function (if it exists, otherwise it rejects) in $poly(2^n)$ time.
\end{clm}
\begin{proof}
we infer from Corollary \ref{cor:Main} that if $f$ is $\odot$-decomposable (where $\odot \in\{\vee,\wedge,\oplus\}$) then every $ROF_\oplus$ for it must have a top gate $\odot$. If $\odot \in \{\vee, \wedge\}$ then the decomposition is unique, therefore the algorithm produces a correct minimal formula for $f$ iff it does so for each of its factors. Therefore the correctness follows by induction on the number of variables. If $f$ is $\oplus$-decomposable, then the decomposition is unique only up to distribution of constants among the factors, or namely, distribution of negations. However, it is easy to see that negation does not alter the size of the formula thanks to De-Morgan laws, while negating an $\oplus$ gate is simply negating one of its inputs. Therefore preforming the recursion over the functions in the canonical $\oplus$-decomposition results in a formula of the same size as in every other decomposition, and the claim follows in this case too. \\ 
\indent The complexity may be seen as polynomial using a simple analysis of the recursion tree. Since any node of the recursion corresponds to a variable (or to a non-decomposable function) there are at most $n$ nodes. A known combinatorial claim (provable using straightforward induction) is that a tree with $n$ nodes, such that any inner node has at least 2 sons, has at most $2n$ nodes. Every node of the recursion tree requires polynomial, time , and therefore the overall complexity is $poly(2^n)$.
\end{proof}

\subsection{Boolean and Arithmetic Read Once Formulas} \label{sctn:OurROF}

The following is a corollary of Theorem \ref{thm:Learning}:

\begin{corollary}
Boolean and arithmetic ROFs have an efficient truth-table minimization algorithm.
\end{corollary}

\begin{proof}
Let $A,B$ be the algorithms mentioned in Theorem \ref{thm:Learning}. Since any function has at most one (boolean or arithmetic) $ROF$ representation, the minimization algorithm $B$ may just output its input. For the complexity of $A$ for boolean ROFs, the learning algorithm of \cite{AHK89} properly and exactly learns a read-once function using $O(n^4)$ time, $O(n^3)$ membership queries and $O(n)$ equivalence queries. Since simulating a membership query requires $O(1)$, and simulating an equivalence query requires $O(nN)$ time, the overall complexity of $A$ is $O(nN)\in poly(N)$. Similarly, the algorithm of \cite{BHH92} uses $O(n^5)$ membership and $n$ equivalence queries, therefore $A$ takes $O(nN)$ time also in the case of arithmetic ROFs.
\end{proof}

We note that the algorithm for read-once formulas over $\{\vee,\wedge,\oplus\}$ presented earlier may also be used to construct boolean read-once functions by only checking $\wedge$ and $\vee$-decomposability. Since every function that has a boolean $ROF$ representation is a unate function, we may use the unate DNF/CNF representation to check decomposability in $O(N)$ time (see Observation \ref{obsrv:HowToFindDecompositions}). This would yield an $O(nN)$ minimization algorithm, which is identical to what we get by using results from computational learning theory.
Moreover, in the case of arithmetic formulas over \textit{any} field, Theorem \ref{thm:mainWedgeXor} may be generalized to any finite field, by replacing $\oplus$ with addition and $\wedge$ with multiplication. The proof for this generalized version of Theorem \ref{thm:mainWedgeXor} will be identical, since every read-once formula over any field represents a multilinear polynomial, and no restriction on the size of the field exists during its proof. \\ \indent However, factoring a multilinear polynomial is a costly operation (albeit polynomial), and the resulting algorithm falls way behind the one which uses computational learning theory. Therefore, we omit the details.

\subsection{Minimization of boolean $ROF$ with Costly Negation gates} \label{sctn:CostlyBoolean}

It is common to disregard the cost of negation gates, and allow the inputs to be in a negated or non-negated form. Since minimization is our concern we consider the more natural case where inputs arrive only in their non-negated form, and negation gates may be placed on every edge (notice that the additional cost of the negation gates may blow up the size of the formula by a multiplicative factor of 2). It is easy to see that in this model the read-once representation is not unique (see Figure \ref{fig:BooleanROFExample}). However, we show how to find a minimal representation in this model under a certain restriction, given a formula in the ordinary model (i.e., where negations are only in the leaves). For convenience we shall assume that the root of the formulas has an outgoing edge.

\begin{figure}
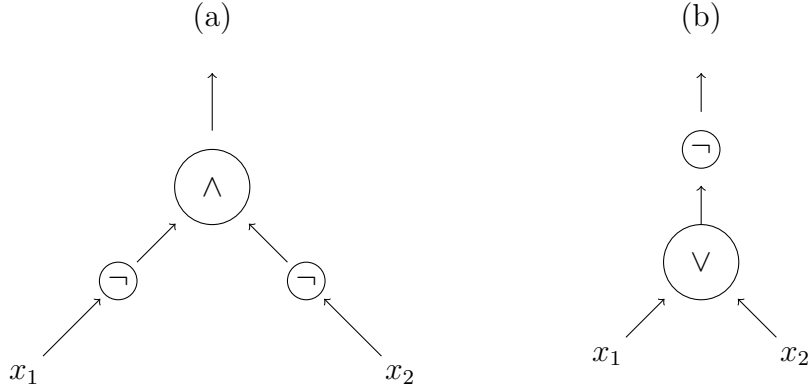
 
\begin{center}
\begin{pgfpicture}{0cm}{0cm}{11.5cm}{6cm}

\pgfputat{\pgfxy(3  ,5.25)}{\pgfbox[center,center]{(a)}}
\pgfputat{\pgfxy(9.5,5.25)}{\pgfbox[center,center]{(b)}}

\pgfsetendarrow{\pgfarrowto}


\pgfcircle[stroke]{\pgfpoint{3cm}{3cm}}{0.5cm}
\pgfputat{\pgfxy(3,3)}{\pgfbox[center,center]{$\wedge$}}

\pgfcircle[stroke]{\pgfpoint{1.75cm}{1.75cm}}{0.25cm}
\pgfputat{\pgfxy(1.75,1.75)}{\pgfbox[center,center]{$\neg$}}

\pgfcircle[stroke]{\pgfpoint{4.25cm}{1.75cm}}{0.25cm}
\pgfputat{\pgfxy(4.25,1.75)}{\pgfbox[center,center]{$\neg$}}

\pgfline{\pgfxy(0.75,0.75)}{\pgfxy(1.5,1.5)}
\pgfline{\pgfxy(2,2)}{\pgfxy(2.5,2.5)}
\pgfline{\pgfxy(5.25,0.75)}{\pgfxy(4.5,1.5)}
\pgfline{\pgfxy(4,2)}{\pgfxy(3.5,2.5)}
\pgfline{\pgfxy(3,3.75)}{\pgfxy(3,4.5)}

\pgfputat{\pgfxy(0.5,0.5)}{\pgfbox[center,center]{$x_1$}}
\pgfputat{\pgfxy(5.5,0.5)}{\pgfbox[center,center]{$x_2$}}


\pgfcircle[stroke]{\pgfpoint{9.5cm}{2cm}}{0.5cm}
\pgfputat{\pgfxy(9.5,2)}{\pgfbox[center,center]{$\vee$}}

\pgfcircle[stroke]{\pgfpoint{9.5cm}{3.5cm}}{0.25cm}
\pgfputat{\pgfxy(9.5,3.5)}{\pgfbox[center,center]{$\neg$}}

\pgfline{\pgfxy(8.5,1)}{\pgfxy(9,1.5)}
\pgfline{\pgfxy(10.5,1)}{\pgfxy(10,1.5)}
\pgfline{\pgfxy(9.5,2.5)}{\pgfxy(9.5,3)}
\pgfline{\pgfxy(9.5,4)}{\pgfxy(9.5,4.5)}

\pgfputat{\pgfxy(8.25,0.75)}{\pgfbox[center,center]{$x_1$}}
\pgfputat{\pgfxy(10.75,0.75)}{\pgfbox[center,center]{$x_2$}}

\end{pgfpicture}
\end{center}
\caption{\label{fig:BooleanROFExample} Read-once formula representation is not unique if negations are not only at the leaves.}
\end{figure}

\begin{defn}
We say that $neg\left(e\right)=1$ if the edge $e$ has a negation gate placed on it and $neg\left(e\right)=0$ otherwise. For a formula $T$ we write $neg_{T}\left(e\right)$ to indicate that $e$ belongs to $T$ and has a negation gate placed on it. We omit the notation $T$ when it is clear from context. \\For an edge $e$, we write $lower\left(e\right)$ to indicate the deeper (i.e., farther from the root) node of $e$, and $upper\left(e\right)$ to indicate its shallower node.
\end{defn}

Denote by $ROF$ the ordinary read-once formula model (where the negation are free and located at the leaves). Unfortunately, we could not find a general minimization algorithm for read-once formula over the basis $\left\{ \wedge,\vee,\neg\right\}$ but only for the following restricted model. 

\begin{defn}
$ROF_{\neg}$ is the class of all read once formulas with gates $\vee, \wedge$ or $\neg$, where the variables appear only in their positive form, and a negation gate does not appear between two gates with labels $\wedge, \vee$. The size of a $ROF_{\neg}$ formula is its number of gates.
\end{defn}

Imposing the additional restriction allows us to find the minimal read-once formula which is \textit{structurally} identical to the formula in the $ROF$ model (i.e., up to negated edges). Removing this restriction allows different formula structures, of whom we lack the tools to find. Although the initial question of adding cost to negation gates is a natural one, the additional restriction we impose may seem unnatural in some sense. However, it may serve as a first step towards a minimization algorithm of read-once formula without this restriction. This direction requires a better understanding of the structural effect that De-Morgan laws inflict upon read-once formulas, and we leave it as an interesting open problem.

\indent Denote by $f^{RO}$ the unique ordinary read-once representation for a function $f$. First, we show that any $ROF_{\neg}$ formula consistent with $f$ is a result of applying some operations on the nodes of $f^{RO}$.

\begin{defn}
For a node $v$ labelled $\vee$ or $\wedge$ define the operator $flip\left(v\right)$ as follows: 
\begin{itemize}
\item If $label\left(v\right)=\vee$ change it to $\wedge$ and vice versa. 
\item For every $e$ such that $v\in e$, do $neg\left(e\right)=\overline{neg\left(e\right)}$.
\end{itemize}
\end{defn}

It is easy to see that the operator $flip\left(v\right)$ does not change the functionality of the formula, thanks to De-Morgan laws. Moreover, we also have that performing $flip(v)$ twice does not change the formula at all.

\begin{thm} \label{thm:ROFFlipOperations}
Let $f^{S}$ be some $ROF_{\neg}$ consistent with $f$. Then there exists a series $\left(v_{1},\ldots,v_{k}\right)$ of nodes such that applying $\left(flip\left(v_{1}\right),\ldots,flip\left(v_{k}\right)\right)$ to $f^{RO}$ results in $f^{S}$. 
\end{thm}

\begin{proof}
Perform the following deterministic algorithm on $f^{S}$: While not all negations are adjacent to leaves: choose an edge $e$ of minimal depth such that $neg\left(e\right)=1$ and preform $flip$ on $lower\left(e\right)$. \\
\indent Since the depth of the edge at hand may only decrease, and since the fact that negated edges cannot appear between $\wedge$ and $\vee$ gates does not allow collapsing of identical gates after a $flip$ operation, this algorithm results in $f^{RO}$. Let $\left(u_{1},\ldots,u_{k}\right)$ be the series of nodes that $flip$ was applied over in this algorithm. It is not hard to see that applying the $flip$ operations in a reversed order on $f^{RO}$ results in $f^{S}$.
\end{proof}

The following theorem we will show that the optimal $ROF_{\neg}$ may be found by applying $flip$ exactly once on some subset of vertices, regardless of their order.

\begin{thm}
Let $f^{S}$ be the optimal $ROF_{\neg}$ consistent with a function $f$. Then there exists $A\sus V\left(f^{RO}\right)$ such that applying $flip\left(v_{i}\right)$ for all $v_{i}\in A$  in any order results in $f^{S}$.
\end{thm}

\begin{proof}
Since the skeleton of the formula (i.e., the formula without any $\neg$ gates and with no labels) of any $ROF_{\neg}$ formula is identical to the skeleton of $f^{RO}$, the size that oughts to be minimized is $\sum_{e}neg\left(e\right)$. Let $\left(v_{1},\ldots,v_{k}\right)$  be the series of nodes that were flipped in the transition from $f^{RO}$ to $f^{S}$, such that $\left|\left\{ v_{1},\ldots,v_{k}\right\} \right|$ is minimal. Notice that the contribution of every edge $e=\left(u,v\right)$ in $f^{S}$ to $\sum_{e}neg_{f^{S}}\left(e\right)$ is $neg_{f^{RO}}\left(e\right)\oplus\left(k_{u}+k_{v}\right)_{\mod2}$, when $k_{v},k_{u}$ is the number of times that the vertices $u,v$ appear in $\left(v_{1},\ldots,v_{k}\right)$. Therefore odd number of appearances of a node in $\left(v_{1},\ldots,v_{k}\right)$ may be reduced to a single appearance, even number of appearences may be erased, and the order does not matter.
\end{proof}

\begin{corollary}
The minimal $ROF_{\neg}$ may be found in $poly\left(2^{n}\right)$ time. 
\end{corollary}

\begin{proof}
Apply the ordinary algorithm to get the $ROF$ representation of $f$. Place a negation gate over every negated variable, and cancel its negation. 
Now preform:

\begin{itemize}
\item For each $S\sus V\left(f^{RO}\right)$.
	\begin{itemize}
		\item Define $t_{S}=f^{RO}$.
		\item For each $v\in S$ apply $flip\left(v\right)$ on $t_{S}$.
	\end{itemize}
\item Output the minimal $t_{S}$
\end{itemize}

The correctness follows from the previous theorems. The complexity is $poly\left(2^{n}\right)$ since $f^{RO}$ has at most $n$ leaves, therefore it has at most $2n$ gates, and the loop iterates at most $2^{2n}$ times. Each iteration may be done in $poly\left(n\right)$ time. 
\end{proof}

\subsection{Minimization of Read once formulas over $\left\{ \neg,\oplus,\vee,\wedge\right\}$ with costly negation gates} \label{sctn:CostlyXOR}

As in the previous section we show a minimization algorithm for read once formulas over the basis $\left\{ \neg,\vee,\wedge,\oplus\right\}$ under a similar restriction. Recall that $ROF_{\oplus}$ is the class of read-once formulas over the basis $\left\{ \wedge,\vee,\oplus\right\}$ when negations are in the leaves (see Definition \ref{defn:ROFoPLUS} in Section \ref{sctn:ROFXOR}). We define the following model, which relates to $ROF_{\oplus}$ as $ROF$ relates to $ROF_{\neg}$.

\begin{defn}
$ROF_{\oplus,\neg}$ is the class of read once formulas with gates $\wedge, \vee, \neg$ or $\oplus$, where inputs arrive only in their positive form, constant are not allowed and no $\neg$ gate is located between $\vee$ and $\wedge$ gates.
\end{defn}

Unlike boolean $ROF$, in $ROF_{\oplus}$, even restricting the negations to appear only in leaves does not imply uniqueness (it is easy to see that $\neg x_1 \oplus x_2$ is equivalent to $x_1 \oplus \neg x_2$). However, after choosing which variables are being negated, we get a unique representation for that choice (see Corollary \ref{cor:UniqueROFXorA}). Since there are at most $2^{n}$ options to choose which variables are being negated, if we manage to find them all, we may apply a theorem similar to the Theorem \ref{thm:ROFFlipOperations} to get that we only need to traverse all subsets of nodes with $\wedge$ or $\vee$ label, perform $flip$ to them and reduce redundant negations around $\oplus$ nodes (as will be explained later), to get the minimal $ROF_{\oplus,\neg}$ representation. \\

\indent As a tool in the minimization algorithm, we define the following model: For $a\in\zo{n}$, a $ROF_{\oplus,a}$ formula is a $ROF_{\oplus}$ model such that a variable $x_{i}$ appears in its negated form iff $a_{i}=1$. This model's relation to $ROF_{\oplus}$ may be seen as similar to the relation between DNF and unate DNF. The uniqueness of $ROF_{\oplus,a}$ for every $a$ will be a corollary of the following theorem:

\begin{thm} \label{thm:UniqueROFnoNegations}
Any function $f$ has at most one $ROF_{\oplus}$ representation without any negations.
\end{thm}

\begin{proof}
From Corollary \ref{cor:Main} we infer that for any function that has a representation in the $ROF_{\oplus}$ model, all consistent representations are of the same depth. Therefore we may define the term ''function of depth $d$'' as a function which has some $ROF_{\oplus}$ of depth $d$, and prove the theorem using induction on depth. For functions of depth 1 the claim is obvious. For a function of depth $d$, if it is $\wedge$ or $\vee$-decomposable, then the decomposition is unique, and the claim follows. If $f$ is $\oplus$-decomposable, let $\bigoplus_{i=1}^{k}f_{i}\oplus c$ be the canonical decomposition. We infer that in any $ROF_{\oplus}$ for $f$, all 2nd level function must be either $f_{i}$ or $\overline{f_{i}}$ for some $i\in\left[k\right]$. However, it is not hard to prove that at most one of $f_{i},\overline{f_{i}}$ has a $ROF_{\oplus}$ representation with no negations, and the claim follows.
\end{proof}

\begin{corollary} \label{cor:UniqueROFXorA}
For every $a\in\zo{n}$, any function $f$ has at most one $ROF_{\oplus,a}$ representation.
\end{corollary}

\begin{proof}
Apply Theorem \ref{thm:UniqueROFnoNegations} on the function $f_a (x) = f\left(x\oplus a\right)$.
\end{proof}

Now, we state a definition and a theorem that will allow us to find the minimal $ROF_{\oplus,\neg}$ for a function, given $ROF_{\oplus,a}$ for every a.

\begin{defn}
Let $T$ be a formula in $ROF_{\oplus,\neg}$ form and $S\sus V\left(T\right)\cap\left\{ v\vert label\left(v\right)\in\left\{ \vee,\wedge\right\} \right\}$. Define $flip\left(S\right)$ as the following operation on $T$:
\begin{itemize}
\item $\forall v\in S$, perform $flip\left(v\right)$.
\item For all $v\in V\left(T\right)$ such that $label\left(v\right)=\oplus$, erase any pair of $\neg$ gates that are adjacent to $v$.
\end{itemize}
Denote the resulting formula by $T_{S}$.
\end{defn}

\begin{thm} \label{thm:ROF_OPLUS_NEG_CORRECTNESS}
Let $f:\bb{F}_{2}^{n}\to\bb{F}_{2}$, and let $ROF_{\oplus,a}\left(f\right)$ be the unique $ROF_{\oplus,a}$ formula for $f$ (if exists). Then the size of the minimal $ROF_{\oplus,\neg}$ for $f$ is: 
\begin{eqnarray*}
\min_{a\in\zo{n}}\left[\min_{S\sus V\left(ROF_{\oplus,a}\left(f\right)\right)\cap\left\{ \mbox{\ensuremath{\vee,\wedge} nodes}\right\} }\left|\left(ROF_{\oplus,a}\left(f\right)\right)_{S}\right|\right]
\end{eqnarray*}
\end{thm}

\begin{proof}
We show that given the minimal $ROF_{\oplus,\neg}$ formula $P$ for $f$, there is a finite series of $flip$ operations on $\vee,\wedge$ nodes that we may apply on $P$ to get the unique $ROF_{\oplus,a}$ for some $a$. Afterwards we show that a formula of the same size as $P$ may be achieved by doing a $flip$ operation on some subset of nodes of the unique formula for that $a$. \\ \indent
Define the following algorithm $A$ on $P$: 

\begin{enumerate}
\item While there exists an inner negated edge (i.e., non-adjacent to a leaf),
	\begin{enumerate}
		\item Choose such an edge $e$ of minimum depth.
		\item If $v=lower\left(e\right)$ is a boolean gate ($\wedge$ or $\vee$) do $flip\left(v\right)$.
		\item If $lower\left(e\right)$ is a $\oplus$ gate, remove the negation from $e$ and move it to one of $lower\left(e\right)$'s sons (say, the 			leftmost one).
	\end{enumerate}
\item End While.
\end{enumerate}
It is easy to see that this procedure ends, and it results in the unique $f_{a}\triangleq ROF_{\oplus,a}\left(f\right)$ formula for some $a$. Now, consider $f_{a}$. We claim that 
\begin{eqnarray*}
\min_{S\sus V\left(f_{a}\right)\cap\left\{ \mbox{\ensuremath{\vee,\wedge} nodes}\right\} }\left|\left(f_{a}\right)_{S}\right|=\left|P\right|.
\end{eqnarray*}
The direction $\ge$ is trivial since $P$ is minimal. For the other direction, notice first that due to the restriction that no $\neg$ gate is located between $\wedge$ and $\vee$ gates, no collapse of gates is possible during the execution of $A\left(P\right)$, and therefore the \textit{skeleton} (i.e., the tree that is resulted by erasing labels and negation gates) of $f_{a}$ and $P$ is identical. Second, let $K$ be the series of boolean gates that were flipped during the execution of $A$ on $P$ (which resulted in $f_{a}$), and let $\overline{K}$ be the set of elements that appear odd number of times in $K$. Observe that:

\begin{enumerate}

\item For any edge $e$ that contains no $\oplus$ node, we know (as in previous section) that $neg_{\left(f_{a}\right)_{\overline{K}}}\left(e\right)$ depends only in the parity of the number of appearances of $e$'s edges in $K$. Therefore $neg_{\left(f_{a}\right)_{\overline{K}}}\left(e\right)=neg_{P}\left(e\right)$. 

\item Let $v$ be any node of $P$ with $label\left(v\right)=\oplus$, we have that the contribution of \textit{all} edges adjacent to it in $A\left(P\right)=f_{a}$ is: (when $u_{e}$ is the other non-$\oplus$ vertex of $e$, and $K_{u_{e}}$ is its number of appearances in $K$)
\begin{eqnarray*}
\bigoplus_{e\vert v\in e}\left[neg_{P}\left(e\right)\oplus\left(K_{u_{e}}\right)_{mod2}\right]=\left(\bigoplus_{e\vert v\in e}neg_{P}\left(e\right)\right)\oplus\left(\bigoplus_{u_{e}\vert v\in e}\left(K_{u_{e}}\right)_{mod2}\right)
\end{eqnarray*} 

\end{enumerate}
From 1 and 2 we deduce that applying $flip\left(\overline{K}\right)$ to $f_{a}$ will result in a formula with the same size as $P$, and the claim follows.
\end{proof}

\begin{clm}
The above algorithm finds the minimal $ROF_{\oplus,\neg}$ in $poly(2^n)$ time.
\end{clm}

\begin{proof}
Theorem \ref{thm:ROF_OPLUS_NEG_CORRECTNESS} allows us to devise the following algorithm for finding the minimal $ROF_{\oplus,\neg}$ for a given function: Find all $ROF_{\oplus,a}$  for any possible $a\in\zo{n}$. For each of them traverse all subsets $S$ of $\wedge,\vee$ gates and apply $flip\left(S\right)$ to $f_{a}$. Choose the minimal representation that is received along the way. Its correctness follows immediately from Theorem \ref{thm:ROF_OPLUS_NEG_CORRECTNESS}. \\
\indent Now, in order to find all $ROF_{\oplus,a}$, we devise Algorithm \ref{MinimizeROFoPLUSnoNeg} for $ROF_{\oplus}$ with no negations, and use it for $f(x\oplus a)$ for all $a\in\zo{n}$. The only difference between this algorithm and the $ROF_{\oplus}$ algorithm presented earlier, is in dealing with $\oplus$-decomposability. We already showed that if $f$ is $\oplus$-decomposable with factors $\{f_i\}_{i=1}^{k}$ then in every $ROF_\oplus$ representation of $f$ the functions at the second level are either from the set $\{f_i\}_{i=1}^{k}$ or from the set $\{\overline{f_i}\}_{i=1}^{k}$. We also mentioned that at most one of $f_i,\overline{f_i}$ has a negation-free representation. Therefore it suffices to traverse all $i=1,\ldots,k$ and check if either $f_i,\overline{f_i}$ have a negation-free representation, and return the proper formula iff it indeed represents $f$.
\end{proof}

\begin{algorithm}
\DontPrintSemicolon
Construct $f^{ML}$, the multilinear representation of $f$.\;

\lIf {$f$ is the identity function over a single variable $x_i$}{\Return $x_i$.}\\
\lElse {Reject.}

\If {$f$ is $\wedge$-decomposable as $f=\bigwedge_{i=1}^k f_i$} {\Return $\bigwedge _{i=1}^{k} MinimizeROF_{\oplus}NoNeg(T_{f_i})$.}

\If {$f$ is $\vee$-decomposable as $f=\bigvee_{i=1}^k f_i$} {\Return $\bigvee _{i=1}^{k} MinimizeROF_{\oplus}NoNeg(T_{f_i})$.}

\If {$f$ is $\oplus$-decomposable as $f=[\bigoplus_{i=1}^k f_i]\oplus c$} { 
	\For{$i=1,\ldots,k$} {
		$MinimizeROF_{\oplus}NoNeg(T_{f_i}\oplus b_j)$ for $b_j \in \zo{}$.\;
		\lIf {Both calls rejected} {Reject.}\\
		\lElse {Let $b_{j(i)}$ be the accepted call and $T_i$ the returned formula.}\\
	}
	\lIf {$\bigoplus_{i=1}^k b_{j(i)}= c$} {\Return $\bigoplus _{i=1}^{k} T_i$.}\\
	\lElse {Reject.}
}

Reject.\;

\caption{$MinimizeROF_{\oplus}NoNeg(T_{f})$} \label{MinimizeROFoPLUSnoNeg}
\end{algorithm}

\subsection{Unate Formulas of the Second Order} \label{sctn:UF2}

Since truth-table minimization of depth 2 unate formulas is easy (see Section \ref{sctn:monDNF}), the algorithm mentioned in section \ref{sctn:OurROF} naturally extends to a certain type of unate formulas, which as far as we know do not exist in the current literature. 

The idea behind the extension of the algorithm is simple: at every decomposition step, we choose the minimal representation between the result of the recursive calls, and the representation as a unate DNF/CNF. Since adjacent identical gates may be collapsed together, this model requires some subtle definition, which could be regarded as a natural extension of the ordinary $ROF$ model.

\begin{defn} \label{defn:InducedFormula}
Let $U$ be a set of input nodes in a formula $\mu$, with a lowest common ancestor u. The \textit{sub-formula induced by $U$} is the sub-graph that is rooted at $u$ and contains exactly all 2nd level sub-formulas of $u$ that contain a variable from $U$ (see Figure \ref{fig:InducedSubformula}).
\end{defn}

\begin{figure}
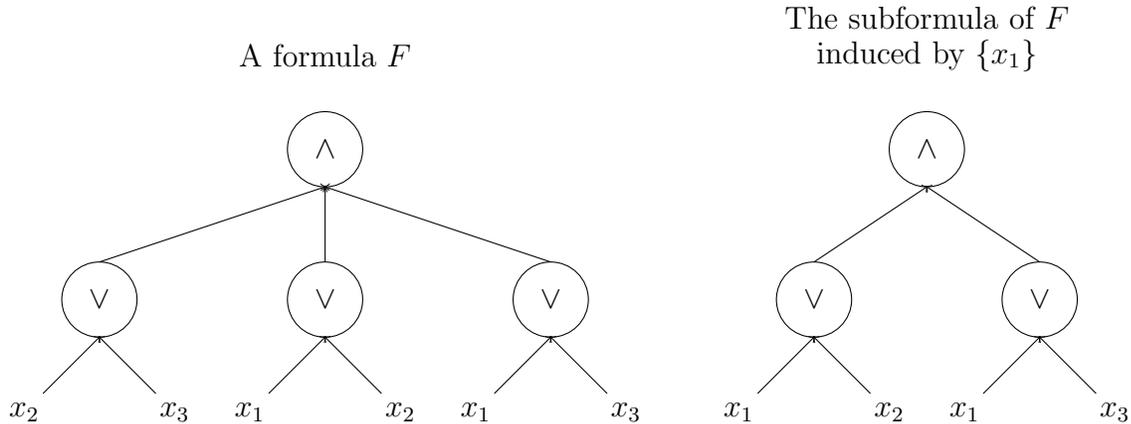
 
\begin{center}
\begin{pgfpicture}{0cm}{0cm}{17cm}{7cm}

\pgfputat{\pgfxy(4.5,5.25) }{\pgfbox[center,center]{A formula $F$}}
\pgfputat{\pgfxy(12.5,5.75)}{\pgfbox[center,center]{The subformula of $F$}}
\pgfputat{\pgfxy(12.5,5.25)}{\pgfbox[center,center]{induced by $\{x_1\}$}}

\pgfsetendarrow{\pgfarrowto}

\pgfcircle[stroke]{\pgfpoint{1.5cm}{2cm}}{0.5cm}
\pgfputat{\pgfxy(1.5,2)}{\pgfbox[center,center]{$\vee$}}

\pgfcircle[stroke]{\pgfpoint{4.5cm}{2cm}}{0.5cm}
\pgfputat{\pgfxy(4.5,2)}{\pgfbox[center,center]{$\vee$}}

\pgfcircle[stroke]{\pgfpoint{7.5cm}{2cm}}{0.5cm}
\pgfputat{\pgfxy(7.5,2)}{\pgfbox[center,center]{$\vee$}}

\pgfcircle[stroke]{\pgfpoint{11cm}{2cm}}{0.5cm}
\pgfputat{\pgfxy(11,2)}{\pgfbox[center,center]{$\vee$}}

\pgfcircle[stroke]{\pgfpoint{14cm}{2cm}}{0.5cm}
\pgfputat{\pgfxy(14,2)}{\pgfbox[center,center]{$\vee$}}

\pgfcircle[stroke]{\pgfpoint{4.5cm}{4cm}}{0.5cm}
\pgfputat{\pgfxy(4.5,4)}{\pgfbox[center,center]{$\wedge$}}

\pgfcircle[stroke]{\pgfpoint{12.5cm}{4cm}}{0.5cm}
\pgfputat{\pgfxy(12.5,4)}{\pgfbox[center,center]{$\wedge$}}

\pgfline{\pgfxy(0.75,0.75)}{\pgfxy(1.5,1.5)}
\pgfline{\pgfxy(2.25,0.75)}{\pgfxy(1.5,1.5)}
\pgfline{\pgfxy(3.75,0.75)}{\pgfxy(4.5,1.5)}
\pgfline{\pgfxy(5.25,0.75)}{\pgfxy(4.5,1.5)}
\pgfline{\pgfxy(6.75,0.75)}{\pgfxy(7.5,1.5)}
\pgfline{\pgfxy(8.25,0.75)}{\pgfxy(7.5,1.5)}
\pgfline{\pgfxy(10.25,0.75)}{\pgfxy(11,1.5)}
\pgfline{\pgfxy(11.75,0.75)}{\pgfxy(11,1.5)}
\pgfline{\pgfxy(13.25,0.75)}{\pgfxy(14,1.5)}
\pgfline{\pgfxy(14.75,0.75)}{\pgfxy(14,1.5)}

\pgfline{\pgfxy(1.5,2.5)}{\pgfxy(4.5,3.5)}
\pgfline{\pgfxy(4.5,2.5)}{\pgfxy(4.5,3.5)}
\pgfline{\pgfxy(7.5,2.5)}{\pgfxy(4.5,3.5)}

\pgfline{\pgfxy(11,2.5)}{\pgfxy(12.5,3.5)}
\pgfline{\pgfxy(14,2.5)}{\pgfxy(12.5,3.5)}

\pgfputat{\pgfxy(0.5,0.5)}{\pgfbox[center,center]{$x_2$}}
\pgfputat{\pgfxy(2.5,0.5)}{\pgfbox[center,center]{$x_3$}}
\pgfputat{\pgfxy(3.5,0.5)}{\pgfbox[center,center]{$x_1$}}
\pgfputat{\pgfxy(5.5,0.5)}{\pgfbox[center,center]{$x_2$}}
\pgfputat{\pgfxy(6.5,0.5)}{\pgfbox[center,center]{$x_1$}}
\pgfputat{\pgfxy(8.5,0.5)}{\pgfbox[center,center]{$x_3$}}
\pgfputat{\pgfxy(10,0.5)}{\pgfbox[center,center]{$x_1$}}
\pgfputat{\pgfxy(12,0.5)}{\pgfbox[center,center]{$x_2$}}
\pgfputat{\pgfxy(13,0.5)}{\pgfbox[center,center]{$x_1$}}
\pgfputat{\pgfxy(15,0.5)}{\pgfbox[center,center]{$x_3$}}

\end{pgfpicture}
\end{center}
\caption{\label{fig:InducedSubformula} Induced subformula.}
\end{figure}

Now define the following model:
\begin{defn} \label{defn:UFK}
A unate formula of order $k$ is a unate formula (over the basis $\{\wedge, \vee\}$, when negations are at the leaves) such that if $U_{i}$  is the set of all input nodes labelled by $x_{i}$ (or $\overline{x_{i}}$), then the sub-formula induced by $U_{i}$ has depth at most $k$. For simplicity, we do not allow adjacent identical gates. This model will be denoted $UF_{k}$. The size of such a formula is the number of leaves. For a boolean function $f$ denote the size of the smallest formula in this model which is consistent with $f$ by $L_{k}\left(f\right)$.
\end{defn}

\begin{remark}
Notice that $UF_{1}$ is the class of boolean read-once formulas and $UF_{n}$ (when $n$ is the number of variables) is the class of all unate functions. In \cite{V12}, arithmetic read-once formulas that bare some single-variable polynomials in the leaves were considered. Our model ($UF_{k}$) may be similarly considered as a boolean read-once formula with some unate formula of depth $\le k$ at the leaves.
\end{remark}

As mentioned in Observation \ref{obsrv:HowToFindDecompositions}, there are simple algorithms for finding the $\wedge$ and $\vee$-decompositions of unate functions in $poly(2^n)$ time. Denote by $find_\wedge(T_f)$ and $find_\vee(T_f)$ the algorithms that receive a truth-table of a function and return the truth-tables of its factors, or return $T_f$ if $f$ indecomposable. In addition, let $MinUnateCNF(T_f)$, $MinUnateDNF(T_f)$ be the minimization algorithms for unate DNF/CNF mentioned in Section \ref{sctn:monDNF}. Notice that the definition of size for whom those algorithm were made is different (number of terms / clauses rather than number of leaves), but in the unate setting it is not hard to prove that a formula is minimal according to one definition iff it is minimal according to the other.
We present algorithm \ref{alg:MinimizeUF2} for minimization of unate formulas of order 2.

\begin{algorithm} \label{alg:MinimizeUF2}
\DontPrintSemicolon
\If {$f$ is a function over 1 variable} {\Return the proper formula.}

$(T_{g_i})_{i=1}^{k} \triangleq find_\wedge(T_f)$.\;
Construct an empty tree $T_{\wedge}$.\;

\If {$k\ne 1$} {
	Add a root $v$ labelled $\wedge$ to $T_{\wedge}$.\;
	\For {$i=1,\ldots, k$} {
		$F_i = MinimizeUF_{2}(T_{g_i})$.\;
		Add $F_i$ to $T_{\wedge}$. If $F_i$ has a top gate $\wedge$, use $v$ instead.\;
	}
}

$(S_{h_i})_{i=1}^{t} \triangleq find_\vee(T_f)$.\;
Construct an empty tree $T_{\vee}$.\;

\If {$t\ne 1$} {
	Add a root $u$ labelled $\vee$ to $T_{\vee}$.\;
	\For {$i=1,\ldots, t$} {
		$K_i = MinimizeUF_{2}(T_{h_i})$.\;
		Add $K_i$ to $T_{\vee}$. If $S_i$ has a top gate $\vee$, use $u$ instead.\;
	}
}

$T_{DNF} \triangleq MinUnateDNF(T_f)$.\;
$T_{CNF} \triangleq MinUnateCNF(T_f)$.\;

\lIf {$T_\vee, T_\wedge, T_{DNF}, T_{CNF}$ are empty} {Reject.}\\
\lElse {\Return the minimal tree among them.}

\caption{$MinimizeUF_{2}(T_{f})$} \label{alg:MinimizeUF2}
\end{algorithm}

Denote by $uDNF_{s}\left(f\right), uCNF_{s}\left(f\right)$ the size of  the minimal $unateDNF, unateCNF$ for $f$, when size is defined to be the number of leaves. 
The following Theorem will show that Algorithm \ref{alg:MinimizeUF2} returns the minimal $UF_2$ representation of the given function in polynomial time.

\begin{thm} \label{thm:UF2Correctness}
Let $f$ be some unate boolean function. 
\begin{enumerate}
\item If $f$ is $\wedge$ or $\vee$-decomposable with the factors $\left\{ f_{i}\right\} _{i=1}^{k}$, then\\ $L_{2}\left(f\right)=\min\left\{ \sum_{i=1}^{k}L_{2}\left(f_{i}\right),uDNF_{s}\left(f\right),uCNF_{s}\left(f\right)\right\} 
 $.

\item If $f$ is indecomposable, then $L_{2}\left(f\right)=\min\left\{ uDNF_{s}\left(f\right),uCNF_{s}\left(f\right)\right\}$.
\end{enumerate}
\end{thm}

To prove this theorem we will need the following direct product lemma:

\begin{lemma} \label{lemma:DirectProduct}
Let $\left\{ f_{i}\left(X_{i}\right)\right\} _{i=1}^{k}$ be variable disjoint unate boolean functions. Then:
\begin{enumerate}
\item $uDNF_{s}\left(\bigwedge_{i=1}^{k}f_{i}\right)\ge\sum_{i=1}^{k}uDNF_{s}\left(f_{i}\right)$

\item $uCNF_{s}\left(\bigwedge_{i=1}^{k}f_{i}\right)=\sum_{i=1}^{k}uCNF_{s}\left(f_{i}\right)$

\end{enumerate}
\end{lemma}

\begin{proof}
For part 1, we use induction on $k$. For $k=1$ there is nothing to prove. Now let $k$ be arbitrary, and assume correctness up to $k-1$. Let $\phi$ be some unate DNF for $\bigwedge_{i=1}^{k}f_{i}\left(X_{i}\right)$. Let $\ga_{1}$ be some assignment on $X_{1}$ such that $f_{1}\left(\ga_{1}\right)=1$ (we assume that such an assignment exists since otherwise, $\bigwedge_{i=1}^{k}f_{i}$ is the constant 0 function, and there is nothing to prove). Denote by $\phi\vert_{\ga_{1}}$ the formula resulting from replacing each leaf $x_{i}\in X_{1}$ in $\phi$ by the corresponding value in $\ga_{1}$, and omitting the leaf according to the assigned value. Obviously $\phi\vert_{\ga_{1}}$ is consistent with $\bigwedge_{i=2}^{k}f_{i}$, therefore by the induction hypothesis we have that $\left|\phi\vert_{\ga_{1}}\right|\ge\sum_{i=2}^{k}uDNF_{s}\left(f_{i}\right)$. Now define $A_{i}$ as the set of leaves in $\phi$ baring variables from $X_{i}$. Since $\left\{ X_{i}\right\} _{i=1}^{k}$ are mutually disjoint, it is clear that
$\left|\phi\right|=\sum_{i=1}^{k}\left|A_{i}\right|$ and $\left|\phi\vert_{\ga_{1}}\right|\le\left|\phi\right|-\left|A_{1}\right|$ since we removed from $\phi$ all leaves that were labelled by a variable from $X_{1}$, and by doing so we may have removed additional leaves not in $A_{1}$. Thus: $\left|\phi\right|\ge\sum_{i=2}^{k}uDNF_{s}\left(f_{i}\right)+\left|A_{1}\right|$. \\
\indent We now claim that $\forall i\in\left[k\right]$, we have $\left|A_{i}\right|\ge uDNF_{size}\left(f_{i}\right)$. \underline{Proof}: assume for contradiction that $\left|A_{i}\right|<uDNF_{s}\left(f_{i}\right)$ for some $i$. Let $\ga^{-i}$ be an assignment to all variables in $X_{j}$ for all $j\ne i$  such that $f_{j}\left(\ga^{-i}\right)=1$ (we assume $\ga^{-i}$ exists as we do for $\ga_{1}$). Clearly, $\phi\vert_{\ga^{-i}}$
is consistent with $f_{i}$, and $\left|\phi\vert_{\ga^{-i}}\right|\le\left|A_{i}\right|$
 since we remove all $A_{j}$ for $j\ne i$. We get that $\phi\vert_{\ga^{-i}}$ is a unate DNF consistent with $f_{i}$, albeit $\left|\phi\vert_{\ga^{-i}}\right|<uDNF_{size}\left(f_{i}\right)$, a contradiction. Using this claim we get $\left|\phi\right|\ge\sum_{i=1}^{k}uDNF_{s}\left(f_{i}\right)$ Which finishes part 1. \\
\indent  As for part 2, the direction $uCNF_{s}\left(\bigwedge_{i=1}^{k}f_{i}\right)\le\sum_{i=1}^{k}uCNF_{s}\left(f_{i}\right)$
is easy. We may just take an $\wedge$ between all uCNF representations of the $f_{i}$'s to get QED. For the other direction, we may repeat the proof of part 1 of this lemma, considering $uCNF_{s}$ instead of $uDNF_{s}$.
\end{proof}

We now turn to prove the correctness of the algorithm.

\begin{proof}
(Of Theorem \ref{thm:UF2Correctness}) For the first part, the direction $\le$ is straightforward - take the minimal representation among the $uDNF$, $uCNF$ and $\wedge$ of the minimal $UF_{2}$ representations of $f$'s factors. For the direction $\ge$ we shall prove the claim by induction on the number of variables of $f$. For $f$ over 1 variable there is nothing to prove. Let $f$ be over 2 variables. For part 1, $f$'s factors must be $f_{1}\left(x_{1}\right),f_{2}\left(x_{2}\right)$. The direction $\le$ is easy. For the direction $\ge$, let $\phi$ be some minimal $UF_{2}$ for $f$. If $\phi$ is of depth $\le 2$, we are done. Furthermore, it is easy to see that there are no $UF_{2}$ for a function over 2 variables of depth $\ge 3$: assume for contradiction that $\phi$ is of depth $\ge 3$. We may conclude that one of its 2nd level sub-formulas is of depth 2, and depends only on (say) $x_{1}$. As far as unate formulas are considered,  there are no minimal depth 2 formulas that depends on 1 variable. \\
\indent Now, let $f$ be a function over any number of variables, let $f=\bigwedge _{i=1}^k f_i (X_i)$ be the maximal $\wedge$-decomposition of $f$ and let $\mu$ be some minimal $UF_{2}$ formula for $f$. If $\mu$ is of depth $>2$, according to the definition of $UF_{2}$, it induces some variable decomposition. Therefore if $\mbox{topgate}\left(\mu\right)=\vee$ we would get an $\vee$-decomposition, in contradiction with Corollary \ref{cor:Main}. Therefore we either have that $\mu$ is of depth $\ge 3$ and has top-gate $\wedge$, or it is a unate DNF/CNF. \\
\indent If it is a unate DNF or a unate CNF, we are done. Else, we may write $\mu$ as a $\wedge$ of variable disjoint functions in the following way: let $\{f_1,\ldots , f_t\}$ be the 2nd level sub-formulas of $\mu$ of depth 1. In order to present $\mu$ as a conjunction of variable disjoint factors, we cluster $\{f_1,\ldots , f_t\}$ into classes such that the variables of each class are disjoint. For every class define $\mu_i$ to be the $\wedge$ of all formulas in the class. Notice that each $\mu_i$ is of depth 1 or 2. Let $D_1$ of the set of indices of those $\mu_i$s. Now let $\{g_1,\ldots , g_s\}$ be the 2nd level sub-formulas of $\mu$ of depth $\ge 2$. From the definition of $UF_2$ it is clear that they are variable disjoint, since if they are not, there is a variable $x_{i}$ such that $U_{x_{i}}$ induces a formula of depth $\ge 3$. For convenience of notation let $\{\mu_i\}_{i\in D_2}$ be the set  $\{g_1,\ldots , g_s\}$, where $D_1 \cap D_2 = \es$. We get:
\begin{eqnarray*}
\mu=\bigwedge_{i\in D_{1}}\mu_{i}\wedge\bigwedge_{j\in D_{2}}\mu_{j}
\end{eqnarray*}
and all factors are variable disjoint. Therefore $\mu$ induces some $\wedge$-decomposition of $f$. According to Corollary \ref{cor:UniquenessOfDecomposition} this decomposition may refined in order to achieve the maximal decomposition. I.e., $\forall i\in D_{1}\cup D_{2},\exists I_{i}\sus\left[k\right],\mu_{i}=\bigwedge_{j\in I_{i}}f_{i}\left(X_{i}\right)$ Such that $\left\{ I_{i}\right\} _{i\in D_{1}\cup D_{2}}$ is a partition of $\left[k\right]$. Now distinguish between the following cases:
\begin{enumerate}
\item $i\in D_{1}$. Notice that in this case, $\mu_{i}$ is represented in its unate CNF form (perhaps as a single clause), using $\mbox{topgate}\left(\mu\right)$. Therefore according to the minimality of $\mu$, lemma \ref{lemma:DirectProduct}, and the induction hypothesis, $L_{2}\left(\mu_{i}\right)	=	uCNF_{s}\left(\bigwedge_{j\in I_{i}}f_{j}\right)
	=	\sum_{j\in I_{i}}uCNF_{s}\left(f_{j}\right)$. Moreover, by the induction hypothesis $\forall j\in I_{i}$ we have $uCNF_{s}\left(f_{j}\right)\ge L_{2}\left(f_{j}\right)$. Hence, $L_{2}\left(\mu_{i}\right)\ge\sum_{j\in I_{i}}L_{2}\left(f_{j}\right)$.

\item $i\in D_{2}$. Since identical adjacent gates are not allowed we know that $\mbox{topgate}\left(\mu_{i}\right)=\vee$. Therefore we have two sub-cases -
	\begin{enumerate}
		\item $\mu_{i}$ is of depth 2, i.e. a unate DNF. Again, by lemma \ref{lemma:DirectProduct}, the minimality of $\mu$ and the induction hypothesis, $L_{2}\left(\mu_{i}\right)	=	uDNF_{s}\left(\bigwedge_{j\in I_{i}}f_{j}\right) \ge	\sum_{j\in I_{i}}uDNF_{s}\left(f_{j}\right)
	\ge	\sum_{j\in I_{i}}L_{2}\left(f_{j}\right)$.

		\item $\mu_{i}$ is of depth larger than 2. In this case, by the induction hypothesis, it is immediate that $L_{2}\left(\mu_{i}\right)=L_{2}\left(\bigwedge_{j\in I_{i}}f_{j}\right)=\sum_{j\in I_{i}}L_{2}\left(f_{j}\right)$.
	\end{enumerate}

\end{enumerate}

Therefore, since $\left\{ I_{i}\right\} _{i\in D_{1}\cup D_{2}}$ is a partition of $\left[k\right]$, we have $\left|\mu\right|	=	\sum_{i\in D_{1}\cup D_{2}}\left|\mu_{i}\right|
	\ge	\sum_{i=1}^{k}L_{2}\left(f_{i}\right)$ as needed.\\
\indent If the function $f$ is $\vee$-decomposable, we observe these simple facts:

\begin{enumerate}
\item $\overline{f}$ is $\wedge$-decomposable, and its factors are $\left\{ \overline{f_{i}}\right\} _{i=1}^{k}$.

\item For every function $g$, $L_{2}\left(g\right)=L_{2}\left(\overline{g}\right)$, since negation does not change the size of the formula.

\item For any function $g$, $uDNF_{s}\left(g\right)=uCNF_{s}\left(\overline{g}\right)$.
\end{enumerate}
Therefore we have 
\begin{eqnarray*}
L_{2}\left(f\right)	& = & L_{2}\left(\overline{f}\right) \\ 
	& = & \min\left\{\sum_{i=1}^{k}L_{2}\left(\overline{f_{i}}\right),uDNF_{s}\left(f\right),uCNF_{s}\left(f\right)\right\} \\
	& = & \min\left\{ \sum_{i=1 }^{k}L_{2}\left(f_{i}\right),uCNF_{s}\left(f\right),uDNF_{s}\left(f\right)\right\} 
\end{eqnarray*}
and the claim follows. \\
\indent For the second part of the theorem (where $f$ is indecomposable), the $\le$ part is obvious. For the $\ge$ part, let $\phi$ be some minimal $UF_{2}$ formula for $f$. If $\phi$ is of depth $\ge 3$, it induces some variable partition, contradicting $f$'s indecomposability. Therefore $\phi$ is either a unate DNF or a unate CNF, and the claim follows.
\end{proof}

\begin{corollary}
Algorithm \ref{alg:MinimizeUF2} finds the minimal $UF_2$ representation of the given function in $poly(2^n)$ time.
\end{corollary}
\begin{proof}
The correctness of algorithm \ref{alg:MinimizeUF2} follows immediately from Theorem \ref{thm:UF2Correctness}, since the algorithm considers all 3 possibilities, and chooses the smallest. \\
\indent To see the polynomial complexity, observe that the recursion tree has at most $n$ leaves, therefore the complexity analysis of Algorithm \ref{alg:MinimizeROFoPLUS} may be applied to reach the same result.
\end{proof}

\subsection{Arithmetic Formulas of the Second Order over $\Ftwo$} \label{sctn:F2A}

In this section we shall see the arithmetic equivalent to $UF_2$ formulas, for which Corollary \ref{cor:Main} also allows us to devise a polynomial truth-table minimization algorithm. Define the following model:

\begin{defn}
$F_2^A$ is the class of read-once formulas over $\Ftwo$ with gates $\oplus, \wedge$ (with constants), where no negations are allowed, and for every variable $x_i$, the set $U_{x_i}$ induces a sub-formula of depth 2 (see Definiton \ref{defn:InducedFormula}). The size an $F_2^A$ formula is its number of leaves, excluding constants.
\end{defn}

We shall see that the minimal $F_2^A$ formula for a given function may be found in $poly(2^n)$ time. The main idea of the minimization algorithm resembles the one of Section \ref{sctn:UF2}. We decompose the function until no decomposition is possible and then apply minimization algorithms for depth 2 formulas. In the unate boolean case, minimization of depth 2 formula may be done easily. We show that in the arithmetic case it may also be done. \\
\indent For $\Sigma_2^A$ formulas the minimization algorithm is trivial over $\Ftwo$, since this model is simply the unique multilinear representation of the function (see Observation \ref{obsrv:HowToFindDecompositions} for an explanation about how to find it). For $\Pi_2^A$ (i.e., a product of linear polynomials) we devise a minimization algorithm. Before presenting the algorithm, we observe:

\begin{obsrv} \label{obsrv:Pi2}
Every function $f$ which has a $\Pi_2^A$ representation is a characteristic function of some affine space of $\Ftwo^n$.
\end{obsrv}

\begin{proof}
Every linear polynomial $P=\sum _{i=1}^n \ga_i x_i + c$ in the first level of a $\Pi_2^A$ formula defines the constraint of the form $f(x)=1 \Rightarrow \left<\ga,x\right>=c+1$ on $x\in\zo{n}$. Hence, the set $f^{-1}(1)$ is the intersection of all constraints in the 1st level, namely, an affine space. 
\end{proof}
Therefore, out of all constraints $(\ga,c)$ that contain $f^{-1}(1)$, we ought to find the smallest independent subset, when the size of the set is the sum of the Hamming weights of the corresponding vectors $\ga$, since we do not count constants.

\begin{clm} 
Algorithm \ref{alg:MinimizePi2A} finds the minimal $\Pi_2^A$ representation of the given function in $poly(2^n)$ time. 
\end{clm}

\begin{proof}
In Algorithm \ref{alg:MinimizePi2A} we first find the set $\{(\ga_i,c_i)\}_{i\in I}$ of all constraints that contain $f^{-1}(1)$, and then use the algorithm $FindBasis$ of \cite{CGH95} as a black box to find the required minimal set. The correctness follows from Observation \ref{obsrv:Pi2} and the correctness of $FindBasis$. To see the polynomial complexity we note that the algorithm $FindBasis$ from \cite{CGH95} has a time bound of $O(n^4 2^n)$ (a complexity analysis does appear in the paper), and therefore Algorithm \ref{alg:MinimizePi2A} is polynomial.
\end{proof}

\begin{algorithm}
\DontPrintSemicolon
$i=0$.\;

\ForAll {$\ga \in \zo{n}$ and $c\in \{0,1\}$} {
	\ForAll{$x \in f^{-1}(1)$} {
		\lIf {$\left<\ga,x\right>\ne c+1$} {Continue to the next $(\ga,c)$.}\\
	}
	Save $(\ga,c)$ as $(\ga_i,c_i)$.\;
	$i = i+1$.\;
}

Find some basis $B$ of $A=\{\ga_j\}_{j=1}^i$.\;
$C=FindBasis(B)$ (w.l.o.g assume that $C$ is a set of indices of vectors in $A$).\;

\Return $\bigwedge_{j\in C} \left[\left(\sum _{i\vert\ga_{ji}=1}x_i\right)+c_j\right]$.
		
\caption{$Minimize\Pi_{2}^A (T_{f})$} \label{alg:MinimizePi2A}
\end{algorithm}

We now present our minimization algorithm for $F_2^A$ formulas (Algorithm \ref{alg:MinimizeF2A}).

\begin{algorithm}
\DontPrintSemicolon
\If {$f$ is over $1$ variable} {\Return the proper formula (one of $\left\{ 0,1,x_{i},x_{i}+1\right\} $).}

Represent $f$ as $P_{f}$, a multilinear polynomial (a.k.a $\Sigma_{2}^{A}$ formula).\;

Define two empty formulas: $F^{MUL}=F^{ADD}=\es$.\;

Factor $P_{f}+b$ efficiently for both $b\in\left\{ 0,1\right\}$.\;
Denote the non-trivial factorization among them (if exists) by $\prod_{i=1}^{k}f_{i}\left(X_{i}\right)$.\;
\If {$k\ne 1$} {
	\ForAll {$i\in\left[k\right]$} {
		Construct the truth-table of $f_{i}$. Denote it $T_{f_{i}}$.\;
		$T_{i}\triangleq\mbox{\mbox{Minimize}}F_{2}^{A}\left(T_{f_{i}}\right)$.\;
	}
	Define $F^{MUL}\triangleq\prod_{i=1}^{k}T_{i}+b$ (if either of the
$T_{i}$s has a top $\cdot$ gate, use $\mbox{topgate}\left(F^{MUL}\right)$
instead).\;
}

Use $P_{f}$ to check $+$-decomposability of $f$ (as in Observation \ref{obsrv:HowToFindDecompositions}),
denote it $P_{f}=c+\sum_{i=1}^{t}g_{i}\left(Y_{i}\right)$ (such that
every $g_{i}$ has no free element).\;

\If {$t\ne 1$} {
	\ForAll {$i\in\left[t\right]$} {
		\label{alg:F2A.LineSplit} Define $ T_{i,0}\triangleq\mbox{Minimize}F_{2}^{A}\left(T_{g_{i}}\right),T_{i,1}\triangleq\mbox{Minimize}F_{2}^{A}\left(T_{\overline{g_{i}}}\right)$.\;
	$T_{i} = \min \{T_{i,0}, T_{i,1}\}$.\;
	$c_i = \arg \min \{T_{i,0}, T_{i,1}\}$.\;
}
Define $F^{ADD}\triangleq\left(c+\sum_{i=1}^{t}c_{i}\right)+\sum_{i=1}^{t}T_{i}$. (if either of the
$T_{i}$s has a top $+$ gate, use $\mbox{topgate}\left(F^{ADD}\right)$
instead).\;
	
}

Define $A_{f}=Minimize\Pi_{2}^{A}\left(T_{f}\right)$.\;
\Return the minimal non-empty formula among $\left\{ F^{MUL},F^{ADD},P_{f},A_{f}\right\} $.\;

\caption{$MinimizeF_{2}^{A}\left(T_{f}\right)$} \label{alg:MinimizeF2A}
\end{algorithm}

The proof resembles the outline of the one in Section \ref{sctn:UF2}. The following lemma resembles lemma \ref{lemma:DirectProduct} and its proof is similar:
\begin{lemma} \label{lem:ArithmeticDirectProduct}
Let $\left\{ f_{i}\left(X_{i}\right)\right\} _{i=1}^{k}$ be boolean variable disjoint functions. Then - 
\begin{enumerate}
\item $\Pi_{2}^{A}\left(\prod_{i=1}^{k}f_{i}\right)=\sum_{i=1}^{k}\Pi_{2}^{A}\left(f_{i}\right)$.
\item $\Sigma_{2}^{A}\left(\prod_{i=1}^{k}f_{i}\right)\ge\sum_{i=1}^{k}\Sigma_{2}^{A}\left(f_{i}\right)$.
\item $\Sigma_{2}^{A}\left(\sum_{i=1}^{k}f_{i}\right)=\sum_{i=1}^{k}\Sigma_{2}^{A}\left(f_{i}\right)$.
\item $\Pi_{2}^{A}\left(\sum_{i=1}^{k}f_{i}\right)\ge\sum_{i=1}^{k}\Pi_{2}^{A}\left(f_{i}\right)$.
\end{enumerate}
\end{lemma}
\begin{proof}
For 1,2 follow the proof of lemma \ref{lemma:DirectProduct}, replacing
$uCNF_{s}$ with $\Pi_{2}^{A}$ and $uDNF_{s}$ with $\Sigma_{2}^{A}$.
3 Follows immediately from the uniqueness of representation of boolean
functions as multilinear polynomials. For 4, let $\phi$ be some minimal $\Pi_{2}^{A}$
formula for $f\triangleq\sum_{i=1}^{k}f_{i}\left(X_{i}\right)$. Let
$A_{i}$ be the set of leaves of $\phi$ labelled by variables from
$X_{i}$. By definition, $\left|\phi\right|=\sum_{i=1}^{k}\left|A_{i}\right|$.
For any $j\in\left[k\right]$, let $\ga^{-j}$ be an assignment to
$\left\{ X_{i}\right\} _{i\ne j}$ such that $\forall i\ne j,f_{i}\left(\ga^{-j}\right)=0$
(such an assignment exists since we may assume w.l.o.g that no $f_{i}$
is the constant 1 function). Let $\phi\vert_{\ga^{-j}}$ be the formula
$\phi$, such that every literal in $\left\{ X_{i}\right\} _{i\ne j}$
is replaced by its corresponding value in $\ga^{-j}$. Clearly, $\phi\vert_{\ga^{-j}}$
is a $\Pi_{2}^{A}$ formula that represents $f_{i}$, therefore 
\[
\left|\phi\vert_{\ga^{-j}}\right|=\left|A_{j}\right|\ge\Pi_{2}^{A}\left(f_{i}\right)
\]
This gives us the immediate conclusion that $\Pi_{2}^{A}\left(\sum_{i=1}^{k}f_{i}\right)=\left|\phi\right|=\sum_{i=1}^{k}\left|A_{i}\right|\ge\sum_{i=1}^{k}\Pi_{2}^{A}\left(f_{i}\right)$.
\end{proof}
We now turn to prove the correctness of Algorithm \ref{alg:MinimizeF2A}:
\begin{thm} \label{thm:F2Acorrectness}
Let $f$ be a boolean function 
\begin{enumerate}
\item If $f$ or $f+1$ is $\cdot$-decomposable with factors $\left\{ f_{i}\right\} _{i=1}^{k}$,
then $L_{2}^{A}\left(f\right)=\min\left\{ \sum_{i=1}^{k}L_{2}^{A}\left(f_{i}\right),\Sigma_{2}^{A}\left(f\right),\Pi_{2}^{A}\left(f\right)\right\} $.
\item If $f$ is $+$-decomposable with factors $\left\{ f_{i}\right\} _{i=1}^{k}$,
then 
\[L_{2}^{A}\left(f\right)=\min\left\{ \sum_{i=1}^{k}\min\left\{ L_{2}^{A}\left(f_{i}\right),L_{2}^{A}\left(f_{i}+1\right)\right\} ,\Sigma_{2}^{A}\left(f\right),\Pi_{2}^{A}\left(f\right)\right\} .
\]
\item If $f$ is indecomposable, then $L_{2}^{A}=\min\left\{ \Sigma_{2}^{A}\left(f\right),\Pi_{2}^{A}\left(f\right)\right\} $.
\end{enumerate}
\end{thm}
\begin{proof}
We prove each part separately:
\begin{enumerate}
\item If $f$ is $\cdot$-decomposable, follow the proof of Theorem \ref{thm:UF2Correctness}
for the case where $f$ is $\wedge$-decomposable, replacing $uCNF_{s}$
with $\Pi_{2}^{A}$, $uDNF_{s}$ with $\Sigma_{2}^{A}$, and the following minor change - if $\mu$ is of depth $>2$, its topgate $v$ may not induce a variable decomposition, and only in the following
case: it may be a $+$-gate with one constant and one non-constant
sub-formulas. However, in this case the non-constant sub-formula is consistent with $f+1$, while Corollary \ref{cor:Main} tells us that it cannot be $\cdot$-decomposable. Hence, the top-gate of the non-constant sub-formula is an addition gate, and it may be collapsed with $v$. In this case we get that $v$ induces a $\oplus$-decomposition of $f$, again, in contradiction to Corollary \ref{cor:Main}. \\
If $f+1$ is $\cdot$-decomposable repeat the same proof with the
above modifications for $f+1$.
\item By induction on the number of variables. The case of $f$ over 2 variables
is easy, by considering all non-redundant $F_{2}^{A}$ formulas with
2 variables. \\
For $f$ over any number of variables $n$, the direction $\le$ is
easy. For $\ge$, let $\phi$ be some minimal $F_{2}^{A}$ for $f$. If $\phi$ is of depth $\le 2$, we are done. If the depth is $\ge 3$ and $\mbox{topgate}\left(\phi\right)=\cdot$, we get that $f$ is $\cdot$-decomposable, contradicting Corollary \ref{cor:Main}. Therefore, as stated in the previous section, we may assume that $v=\mbox{topgate}\left(\phi\right)$
is an addition gate, and it induces a variable partition. 
Let $\left\{ \phi_{i}\right\} _{i=1}^{t}$ be the sub-formulas of $\phi$.
As in the proof of Theorem \ref{thm:UF2Correctness} we may write: 
\[
\phi=\sum_{i\in D_{1}}\phi_{i}\left(X_{i}\right)+\sum_{i\in D_{2}}\phi_{i}\left(X_{i}\right)
\]
And for the same reasons, together with the fact that $\oplus$-decomposition is unique up to distribution of constants, we may write:
\[
\forall i\in D_{1}\cup D_{2},\exists I_{i}\sus\left[k\right],\exists c_{i}\in\left\{ 0,1\right\} ,\phi_{i}=\sum_{j\in I_{i}}f_{i}\left(X_{i}\right)+c_{i}
\]
For a partition $\left\{ I_{i}\right\} _{i\in D_{1}\cup D_{2}}$ of
$\left[k\right]$. Now distinguish between the following cases:

\begin{enumerate}
\item $i\in D_{1}$. In this case, $\phi$ is represented as $\Sigma_{2}^{A}$,
using $\mbox{topgate}\left(\phi\right)$. Using lemma \ref{lem:ArithmeticDirectProduct},
the induction hypothesis and $\phi$'s minimality, we deduce - 
\begin{eqnarray*}
L_{2}\left(\phi_{i}\right) & = & \Sigma_{2}^{A}\left(\phi_{i}\right)\\
 & = & \Sigma_{2}^{A}\left(\sum_{j\in I_{i}}f_{j}+c_{i}\right)\\
 & = & \Sigma_{2}^{A}\left(\sum_{j\in I_{i}}f_{j}\right)\\
 & = & \sum_{j\in I_{i}}\Sigma_{2}^{A}\left(f_{j}\right)\\
 & \ge & \sum_{j\in I_{i}}L_{2}^{A}\left(f_{j}\right)\\
 & \ge & \sum_{j\in I_{i}}\min\left\{ L_{2}^{A}\left(f_{j}\right),L_{2}^{A}\left(f_{j}+1\right)\right\} 
\end{eqnarray*}

\item $i\in D_{2}$. We have that $\mbox{topgate}\left(\phi_{i}\right)=\cdot$. Consider two subcases:

\begin{enumerate}
\item $\phi_{i}$ is a $\Pi_{2}^{A}$. Using the same tools as in the previous
section, we have - 
\begin{eqnarray*}
L_{2}^{A}\left(\phi_{i}\right) & = & \Pi_{2}^{A}\left(\sum_{j\in I_{i}}f_{j}\left(X_{j}\right)+c_{i}\right).
\end{eqnarray*}
In order to get rid of the constant $c_i$, define $g\left(X_{j'}\right)\triangleq f_{j'}\left(X_{j'}\right)+c_{i}$
for some $j'\in I_{i}$, to get - 
\begin{eqnarray*}
L_{2}^{A}\left(\phi_{i}\right) & = & \Pi_{2}^{A}\left(\sum_{j\ne j'\in I_{i}}f_{j}+g\right)\\
 & \ge & \sum_{j\ne j'\in I_{i}}\Pi_{2}^{A}\left(f_{j}\right)+\Pi_{2}^{A}\left(g\right)\\
 & \ge & \sum_{j\ne j'\in I_{i}}L_{2}^{A}\left(f_{j}\right)+L_{2}^{A}\left(g\right)\\
 & \ge & \sum_{j\in I_{i}}\min\left\{ L_{2}^{A}\left(f_{j}\right),L_{2}^{A}\left(f_{j}+1\right)\right\} 
\end{eqnarray*}

\item $\phi_{i}$ is not a $\Pi_{2}^{A}$. In this case it is immediate
from the induction hypothesis that - 
\begin{eqnarray*}
L_{2}^{A}\left(\phi_{i}\right) & = & L_{2}^{A}\left(\sum_{j\in I_{i}}f_{j}\right)\\
 & = & \sum_{j\in I_{i}}\min\left\{ L_{2}^{A}\left(f_{j}\right),L_{2}^{A}\left(f_{j}\right)+1\right\} 
\end{eqnarray*}
\\
We get that, since $\left\{ I_{j}\right\} _{j=1}^{t}$ is a partition
of $\left[k\right]$, we have:
\begin{eqnarray*}
L_{2}^{A}\left(f\right) & = & \left|\phi\right|\\
 & = & \sum_{i=1}^{t}\left|\phi_{i}\right|\\
 & \ge & \sum_{i=1}^{k}\min\left\{ L_{2}^{A}\left(f_{i}\right),L_{2}^{A}\left(f_{i}+1\right)\right\} 
\end{eqnarray*}
And the theorem follows in this case.
\end{enumerate}
\end{enumerate}
\item $f$ cannot have any $F_{2}^{A}$ of depth $\ge 3$, since it induces
a variable decomposition, thus any $F_{2}^{A}$ for $f$ is of depth
at most $2$, and the theorem follows.
\end{enumerate} 
\end{proof} 

\begin{corollary}
Algorithm \ref{thm:F2Acorrectness} finds the minimal $F_2^A$ representation of the given function in $poly(2^n)$ time. 
\end{corollary}
\begin{proof}
This follows easily from Theorem \ref{thm:F2Acorrectness}: the algorithm checks the decomposability of the function, knowing that at most one of the decompositions is possible, and outputs the minimal representation according to the decomposability it found. The complexity analysis is similar to that of Algorithm \ref{alg:MinimizeUF2}.
\end{proof}

\section{Open Problems} \label{sctn:OpenProblems}

There are several known examples in the theory of complexity for problems that inherently depend on some parameter $q$ such that for $q=2$ the problem is easy and for $q=3$ the problem becomes hard ($2SAT$ and $3SAT$, 2 and 3 colorability, computing the rank of a matrix and a tensor of dimension 3, etc.). As may seem from this chapter, the problem of finding the minimal unate formula consistent with a given truth-table may also be one of those problems, when $q$ indicates the depth of the formula. \\
\indent For general (non-unate) formulas, it is known that the corresponding decisional problem is $NP$-complete even for depth 2. As we've shown in Section \ref{sctn:LowerBoundsSigma3} proving similar results about depth-3 formulas may be a very hard task. As for unate formulas, we've shown that the existing algrithms for depth 2 minimization may be applied to get minimization algorithms for wider classes of formulas over several different bases ($UF_2, F_2^A$ formulas) using techniques from the world of $ROF$ minimization. To the best of our knowledge, there is no hardness result regarding the truth-table minimization of general unate formulas.\\
\indent Therefore, the most intriguing gap to be closed is the hardness of $MinUnate\Pi_3, MinUnate\Sigma_3$ and $Min\Sigma_3^A, Min\Pi_3^A$, i.e., of unate and arithmetic formulas of depth 3. On one hand, an efficient algorithm for minimization of such formulas will immediately provide an efficient algorithm for $UF_3$ formulas, and may constitute a step towards a construction of a minimization algorithm for general unate formulas. On the other hand, $NP$-complteness of one of those problems may provide new lower bounds for monotone formulas, as explained in Section \ref{sctn:LowerBoundsSigma3} (this is since the lemma by Valiant presented in that section may also be applied similarly to monotone formulas).\\
\section{Some Proofs}
\subsubsection{Proof of Proposition \ref{prop:2Partitions}}

Let $x_{i}$ for some $i\in\left[n\right]$. We know that $\left\{ X_{i}\right\} _{i=1}^{k},\left\{ X_{i}'\right\} _{i=1}^{k'}$ are partitions, therefore there exist $i_{1}\in\left[k\right],i_{2}\in\left[k'\right]$
such that $x_{i}\in X_{i_{1}}\cap X_{i_{2}}'$. Distinguish between
two cases - 

1. $\exists j_{2}\in\left[k'\right]$ such that $x_{i}\notin X_{j_{2}}'$
(and therefore $j_{2}\ne i_{2}$) and $X_{j_{2}}'\backslash X_{i_{1}}\ne\es$.
In this case let $x_{j}\in X_{j_{2}}'\backslash X_{i_{1}}$. Let $j_{1}\in\left[k\right]$
be such that $x_{j}\in X_{j_{1}}$ (note that $j_{1}\ne i_{1}$ since
$x_{j}\notin X_{i_{1}}$). We get $x_{i}\in X_{i_{1}}\cap X_{i_{2}}'$
and $x_{j}\in X_{j_{1}}\cap X_{j_{2}}'$ for $i_{1}\ne j_{1},i_{2}\ne j_{2}$
as needed.

2. $\forall j\in\left[k'\right]$ such that $x_{i}\notin X_{j}'$
we have that $X_{j}'\sus X_{i_{1}}$. Obviously - $X_{i_{2}}'\nsubseteq X_{i_{1}}$,
since if $X_{i_{2}}'\sus X_{i_{1}}$, we get that $\forall j\in\left[k'\right],X_{j}'\sus X_{i_{1}}$,
and therefore the partition $\left\{ X_{i}\right\} _{i=1}^{k}$ is
trivial. Now instead of choosing the aforementioned $x_{i}$, let
$x_{t}\in X_{i_{2}}'\backslash X_{i_{1}}$. Let $t'\in\left[k\right]$
such that $x_{t}\in X_{t'}$ (notice that $t'\ne i_{1}$, since $x_{t}\notin X_{i_{1}}$).
We get that $x_{t}\in X_{t'}\cap X_{i_{2}}'$. To choose his counterpart,
let $j_{2}\in\left[k'\right]$ be \emph{different} from $i_{2}$.
We know that $x_{i}\notin X_{j_{2}}'$, thus according to the assumption
we get $X_{j_{2}}'\sus X_{i_{1}}$, and any element $x_{j}\in X_{j_{2}}'$
will satisfy $x_{j}\in X_{j_{2}}'\cap X_{i_{1}}$. Since $t\ne j$,
$j_{2}\ne i_{2}$ and $t'\ne i_{1}$, we get QED for the pair $x_{t},x_{j}$. 

\subsubsection{Proof of Proposition \ref{prop:Delta2Not0}}
According to \cite[Lemma 2.1.8]{V12} every multilinear polynomial $P$ depends on $x_{i}$ iff $\frac{\partial P}{\partial x_{i}}\ne 0$. Now, from the fact that $Q$ is $(x_i,x_j)$-decomposable and multi-linear, we infer that it may be written as $Q(X)=Q_1(X_1)\cdot Q_2(X_2)$ such that $X_1 \cap X_2 =\es$ and $x_i \in X_1,x_j \in X_2$. By \cite[Lemma 2.1.9]{V12}, we have that partial derivatives of multi-linear polynomials comply with the ordinary sum and product rules, as ordinary derivative does. We get that:
\[
\frac{\partial Q}{\partial x_i x_j} = \frac{\partial Q_1}{\partial x_i} \cdot \frac{\partial Q_2}{\partial x_j}
\]
To see that $\frac{\partial Q}{\partial x_i x_j}$ is not the zero polynomial, observe that $\frac{\partial Q_1}{\partial x_i}$ and $\frac{\partial Q_2}{\partial x_j}$ are non-zero and variable disjoint - Remark \ref{remark:AllVariables} gives us that $Q$ depends on all its variables, and $x_i$ (resp. $x_j$) appears only in $Q_1$ (resp. $Q_2$). Therefore we may choose assignments $\ga,\gb$ for $X_1,X_2$ respectively such that $\frac{\partial Q_1}{\partial x_i}(\ga)\ne 0$ and $\frac{\partial Q_2}{\partial x_j}(\gb)\ne 0$ and get that the concatenation of $\ga,\gb$ is an assignments on whom $\frac{\partial Q}{\partial x_i x_j}$ does not vanish.

\subsubsection{Addendum to the proof of Theorem \ref{thm:LowerBounds}}

\begin{lemma} \label{lemma:LowerBoundsAddendum}
Let $S:\bb N\to\bb N$ be a monotone function such that $S\left(n\right)=o\left(n^{\e}\right)$ for every $\e>0$. Then $S\left(n\right)^{S\left(n\right)}=o\left(2^{\e n^{\e}\log n}\right)$.
\end{lemma}

\begin{proof}
Using basic calculus (and abusing the notation by writing $S\left(n\right)$
for some differentiable continuous monotone function that agrees with
the original $S\left(n\right)$ over $\bb{N}$) we get:

\begin{eqnarray*}
\lim _{n\to \infty}\frac{S\left(n\right)^{S\left(n\right)}}{2^{\e n^{\e}\log n}} & \overset{_{\mbox{L'Hopital}}}{=} & \lim _{n\to\infty}\frac{S\left(n\right)^{S\left(n\right)}\left[S'\left(n\right)+S'\left(n\right)\log\left(S\left(n\right)\right)\right]}{n^{\e n^{\e}}\left[\e\left(n^{\e-1}\right)+\e\left(n^{\e-1}\right)\log\left(n^{\e}\right)\right]}\\
 & = & \lim _{n\to\infty}\frac{S\left(n\right)^{S\left(n\right)}\left[S'\left(n\right)\left(1+\log\left(S\left(n\right)\right)\right)\right]}{n^{\e n^{\e}}\left[\e\left(n^{\e -1}\right)\left(1+\log\left(n^{\e}\right)\right)\right]}
\end{eqnarray*}
Since $S\left(n\right)=o\left(n^{\e}\right)$, then $\lim_{n\to\infty}\frac{S\left(n\right)}{n^{\e}}=\lim_{n\to\infty}\frac{S'\left(n\right)}{\e n^{\e-1}}=0$,
thus 
\[
=\lim_{n\to\infty}2^{S\left(n\right)\log S\left(n\right)-n^{\e}\log n^{\e}}\cdot\frac{S'\left(n\right)}{\e n^{\e-1}}\cdot\frac{1+\log\left(S\left(n\right)\right)}{1+\log\left(n^{\e}\right)}
\]
Now, since the function $\log$ in order-preserving, we have that
$\frac{\log S\left(n\right)+1}{\log\left(n^{\e}\right)+1}\le c$ for
some $c\in\bb R$ and large enough $n$. Moreover, it is easy to see
that 
\[
S\left(n\right)\log\left(S\left(n\right)\right)-n^{\e}\log n^{\e}\le const
\]
for large enough $n$. Therefore the above limit is a multiplication
of two bounded functions, and one that goes to zero with $n$, therefore
the entire limit is 0, and $S\left(n\right)=o\left(2^{\e n^{\e}\log n}\right)$.
\end{proof}

Therefore if $S\left(n\right)=o\left(n^{\e}\right)$ then 
\begin{eqnarray*}
2^{2c\log n\cdot S\left(n\right)}\cdot S\left(n\right)^{S\left(n\right)+1} & = & S\left(n\right)\cdot n^{2c\cdot S\left(n\right)}\cdot S\left(n\right)^{S\left(n\right)}\\
 & = & o\left(n^{\e}\cdot n^{2c\cdot n^{\e}}\cdot2^{\e n^{\e}\log n}\right)\\
 & = & o\left(2^{\e\log n}\cdot2^{2c\cdot n^{\e}\log n}\cdot2^{\e n^{\e}\log n}\right)\\
 & = & o\left(2^{\left(2c+\e\right)n^{\e}\log n}\right)\\
 & = & o\left(2^{n^{2\e}}\right).
\end{eqnarray*}

\chapter{Pseudorandomness} \label{chapt:Pseudo}
Originally used for cryptographic purposes, pseudorandom functions (PRF) were used several times outside of cryptography to show interesting interdisciplinary results, albeit the lack of unconditional constructions of those so far. E.g., in computational learning theory, the existence of PRFs in a certain class provides a barrier for the ability to efficiently learn that class. In complexity theory, the celebrated result of \cite{RR97} that the existence of PRFs which are strong in some sense implies that circuit lower bounds cannot be proved using ``natural'' arguments. Furthermore, in the field of truth-table minimization, the conjectured existence of PRFs provided inapproximability results for several classes of computational models (in \cite{KW09} for communication protocols, in \cite{AH+06} for $AC^0$ circuits, and Theorem \ref{thm:BP_approx} in this Thesis). Similarly, \cite{KC00} show that if the truth-table minimization of a general boolean circuit is easy, then there are no strong PRFs.

In this chapter we generalize the connection between the existence of PRFs in a certain class, and the truth-table minimization of that class (Section \ref{sctn:Generalization}). Moreover, we generalize the aforementioned result of \cite{KC00} to any complexity class in Section \ref{sctn:NaturalProperties} where our terminology is based on \cite{RR97}. One possible use of Section \ref{sctn:NaturalProperties}, regarding OBDDs (Definition \ref{defn:OBDD}), is presented in Section \ref{sctn:OBDD}.

Before stating our generalizations, we define pseudorandom functions:

\begin{defn} \label{defn:PRF}
Let $\cl{F} = \{F_n\}_{n\in\bb{N}}$ be a function ensemble, such that for all $n$, the set $F_n$ is a set of boolean functions on $n$ variables. $\cl{F}$ is called a \textit{pseudorandom function ensemble} (PRFE), and the functions in it are called \textit{pseudorandom functions} (PRF) if for every probabilistic polynomial-time oracle machine $A$, for every polynomial $p(n)$ and for every large enough $n$, 
\[
\left| \Pr(A^{F_n}(1^n) = 1) - \Pr(A^{H_n}(1^n) = 1) \right| \le \frac{1}{p(n)},
\]
Where $A^{F_n}$ denotes $A$ with an oracle access to a function uniformly chosen from $F_n$ and $A^{H_n}$ denotes $A$ with an oracle access to a uniformly random function on $n$ variables. The probability is taken over the choice of the function and the coin flips of $A$.
\end{defn}

\section{Truth-Table minimization and PRFs, a Generalization} \label{sctn:Generalization}

Theorem \ref{thm:PRFandTTM} generalizes the technique of the following results: \cite{KW09} prove that the minimal complexity of a communication protocol for a function given as truth table is inapproximable up to some constant factor. \cite{AH+06} prove inapproximability of $N^{1-\e}$ (for every $\e>0$) of the minimal size of an $AC^0$ circuit. We give $N^{c}$ hardness of approximation result ($0<c<\frac{1}{2}$) for the minimal branching program size (Theorem \ref{thm:BP_approx}). All three result use the conjectured existence of PRFs in the target model in the following way:

\begin{thm} \label{thm:PRFandTTM}
Let $C$ be a class of computational models. If $C$ satisfies the following
conditions:
\begin{enumerate}
\item $1-\exp\left(-\Omega\left(n\right)\right)$ of the functions on $n$ variables require
a model of size at least $l\left(n\right)$.

\item There is a PRFE in the model $C$ such that the maximal $C$-complexity of a function on $n$ variables in the ensemble is $u\left(n\right)$.

\item There exists $\e >0$ and a function $\ga(n)>0$, such that $u\left(2^{n / \e} \right) \le \frac{1}{\ga (n)^{2}}\cdot l\left(n\right)$ for every large enough $n$.

\end{enumerate}
Then there is no polynomial $\ga (n)$-approximation to the truth-table minimization of the model $C$, i.e., there is no polynomial algorithm $A$ such that $\frac{C(f)}{\ga(n)}\le A(T_f)\le \ga(n) \cdot C(f)$, where $C(f)$ denotes the $C$ complexity of the function $f$.
\end{thm}

\begin{proof}
Assume for contradiction that there exists an algorithm $B$ that given a truth-table $T_f$ of a function $f$ on $n$ variables it $\ga (n)$ approximates the size $C(f)$ of the minimal model in $C$ which is consistent with $f$. Define the following algorithm $A$ with an oracle access to $T$: On $1^n$, the algorithm $A$ defines $m\triangleq \e \log n$, and constructs the truth-table $T_g$ of $g(x) \triangleq  T(x0^{n-m})$ on $m$ variables. It then feeds the truth-table into $B$. If $B$'s output is $\le \ga (m)\cdot u\left( 2^{m/\e}\right)$ then $A$ outputs 1, otherwise it outputs 0. \\
\indent We show that $A$ is a distinguisher for the PRFE. If $T$ is a random oracle, then with probability $1-\exp ^{-\Omega (n)}$:
\[
B(T_g) \ge \frac{1}{\ga(m)} \cdot C(g) \ge \frac{1}{\ga(m)} \cdot l(m).
\]
On the other hand, if $T$ is pseudorandom then:
\[ 
B(T_g) \le \ga (m) \cdot C(g) \le \ga (m) \cdot C(T) \le \ga (m) \cdot u(n)=\ga(m)\cdot u\left( 2^{m/\e}\right).
\]
Since $\ga(m)\cdot u\left( 2^{m/\e}\right) \le l(m) / \ga(m)$ we get that 
\[
\left| \Pr(A^{F_n}(1^n) = 1) - \Pr(A^{H_n}(1^n) = 1) \right| \ge 1-\exp ^{-\Omega (n)},
\]
in contradiction with the definition of PRFE.
\end{proof}

\section{PRFs, Natural Properties and Truth-Table Minimization} \label{sctn:NaturalProperties}
While the result of Section \ref{sctn:Generalization} is usually used to get inapproximability results under cryptographic assumptions, in this section we would like to do the opposite. Namely, to formulate exactly what properties of a PRFE could be proved not to exist under the assumption that some model has an efficient truth-table minimization algorithm. We use the terminology of \cite{RR97} to formulate our result. We also show an example of a possible use of this theorem for the OBDD model (see Definition \ref{defn:OBDD}) which has a truth-table minimization algorithm \cite{FS90}. Unfortunately, this corollary is superseded by a result of \cite{KL01}, which shows a result stronger than ours using communication complexity arguments. We leave the search for other implementations of Theorem \ref{thm:NaturalProperty} as an open problem.

\begin{defn}
Let $\cl{C} = \left\{ C_{n}\right\} _{n\in\bb{N}}$ be a property of boolean functions, where $C_n$ is a set of functions on $n$ variables. For any complexity class $\Gamma$ we say that $\cl{C}$ is a $\Gamma$-natural property if it holds:

\begin{enumerate}
\item The predicate  $f_{n}\in C_{n}$ is computable in $\Gamma$ (when
$f_{n}$ is represented as its truth-table).

\item $\left|C_{n}\right|\ge\delta_{n}\cdot\left|H_{n}\right|$. When $H_n$ is the set of all boolean functions on $n$ variables, and for $\delta_{n}\ge 1 - \frac{1}{p(n)}$ for some polynomial $p$ and large enough $n$.

\end{enumerate}
\end{defn}

\begin{defn}
For a property $\cl{C} = \left\{ C_{n}\right\} _{n\in\bb{N}}$, we call a set of functions $\left\{ \mathcal{F}_{n}\right\} _{n\in\bb{N}}$ ``good'' for $\cl{C}$ if:
\begin{enumerate}
\item for every large enough $n$, and for every $f_{n}\in\mathcal{F}_{n}$,
$f_{n}\notin C_{n}\Rightarrow f_n^{\log n}\notin C_{\log n}$. (when $f_n^{\log n}$ is a sub-function of $f_{n}$ over $\log n$ variables for the assignment,
say, $\star^{\log n}0^{n-\log n}$).

\item $\exists m\in\bb{N}$ such that $\forall n>m,\,\mathcal{F}_{n}\cap C_{n}=\emptyset$.

\end{enumerate}
\end{defn}

\begin{thm} \label{thm:NaturalProperty}
Let $\Gamma$ be a complexity class such that $DTIME\left(O\left(n\right)\right)\sus\Gamma$.
Then if there exists a property $\cl{C}$ which is a $\Gamma$-natural property then there is no PRFE which is good for $C$, and fools algorithms in $\Gamma$.
\end{thm}

\begin{proof}
Let $\cl{C},\Gamma$ be as stated in the Theorem. Assume that $\mathcal{F}=\left\{ \mathcal{F}_{n}\right\} _{n\in\bb{N}}$ is a PRFE which is good for $\cl{C}$. We break the PRFE in the following manner: Let $A$ be an algorithm with an oracle access to either an entirely random function or a random function from the PRFE. $A$
(on input $1^{n}$) will construct a truth-table of a sub-function $g$ of the oracle over $\star^{\log n}0^{n-\log n}$. $A$ will compute the predicate ``$g\in C_{\log n}$''. If it gets an affirmative answer, it outputs 0 (namely, \textit{random}). Otherwise it outputs 1 (namely, \textit{pseudorandom}). 

If the oracle is indeed random, then $g$ has the property $C_{\log n}$ w.p $\delta_{\log n}$. If the oracle is pseudorandom (and ``good''), then for large enough $n$, $f_{n}\notin C_{n}$, and according to $\cl{C}$ being good, $g\notin C_{\log n}$ w.p 1. Computing the truth-table of $g$ takes $O\left(n\right)$ time. The decision $g\in C_{\log n}$ may be done in $\Gamma$ since its size is $O\left(n\right)$, which is the input size of $A$. Therefore $A$ is in $\Gamma$.
\end{proof}

\subsection{Implication to OBDDs} \label{sctn:OBDD}
Let $\cl{C}=\left\{ C_{n}\right\} _{n\in\bb{N}}$, $C_{n}=\left\{ f_{n}:\zo{n}\to\zo{}\vert OBDD\left(f_{n}\right)>2^{\e n}\right\} $
for some fixed $0<\e<1$.

\begin{obsrv}
$\cl{C}$ is $P$-natural.
\end{obsrv}
\begin{proof}
Notice that:
\begin{enumerate}
\item The predicate $f_{n}\in C_{n}$ is computable in $P$ by \cite{FS90}.
\item $\left|C_{n}\right|\ge\delta_{n}\left|F_{n}\right|$, since for any
boolean function $f$, $BP\left(f\right)\le OBDD\left(f\right)$,
and a counting argument for $BP$s by \cite{W00} shows that all but exponentially small fraction of function on $n$ inputs require branching programs of exponential size.
\end{enumerate}\end{proof}

According to Theorem \ref{thm:NaturalProperty} we have that there is no $PRFE$ which is good for $C$. Therefore:

\begin{corollary}
there is no PRFE that can be implemented by OBDDs of size $n^{\e}$ and fool polynomial algorithms.
\end{corollary} 

\begin{proof}
Assume for contradiction that $\mathcal{F}=\left\{ \mathcal{F}_{n}\right\} _{n\in\bb{N}}$
is a PRFE such that $\forall n,\forall f_{n},OBDD\left(f_{n}\right)\le n^{\e}$.
We shall see that $\mathcal{F}$ is good for $C$:
\begin{enumerate}

\item We need to prove that for any large enough $n$, $f_{n}\notin C_{n}\Rightarrow f_{n}^{\log n}\notin C_{\log n}$. We know that $OBDD\left(f_{n}\right)\le n^{\e}$, therefore we may fix all variables $x_{\log n+1},\ldots,x_{n}$ in $f_{n}$'s OBDD to 0, and get an OBDD of size at most $n^{\e}$ that computes $ f_{n}^{\log n}$. I.e., $OBDD\left(f_{n}^{\log n}\right)\le n^{\e}$. In order for $f_{n}^{\log n}$ to be in $C_{\log n}$ we must have $OBDD\left(f_{n}^{\log n}\right)>2^{\e\log n}=n^{\e}$,
therefore $f_{n}^{\log n}\notin C_{\log n}$. 

\item Obviously, we have that for any $n$, $OBDD\left(f_{n}\right)\le n^{\e}$,
therefore for large enough $n$ we have $f_{n}\notin C_{n}$. 
\end{enumerate}
Hence the existence of $\cl{F}$ contradicts Theorem \ref{thm:NaturalProperty}, and the claim follows.
\end{proof}

However, a result by \cite{KL01} shows that no PRFE is implementable by OBDD of \textit{any} polynomial size using communication complexity arguments.

\begin{remark} 
A similar claim may be shown on every model that has a non-trivial truth-table minimization algorithm (see Chapters \ref{chapt:DT}, \ref{chapt:Formulas} for examples).
\end{remark}

\chapter{Branching Programs} \label{chapt:BPs}
\section{Perliminaries and Previous Work} \label{sctn:BPperliminaries}

\begin{defn}
A Branching Program (also known as Binary Decision Diagram) is a generalization of a decision tree in which the underlying graph may not be a tree, but some directed acyclic graph (DAG). Formally, it is a DAG in which every non-terminal node is labelled with a variable
in $\left\{ x_{i}\right\} _{i=1}^{n}$, and has out-degree 2, with
two edges labelled 0 and 1. There are two terminal nodes labelled
0 and 1. A branching program $P$ is said to compute the function
$f$ iff $\forall\overline{a}\in\left\{ 0,1\right\} ^{n}$, the path
that begins at the root, and follows the edges labelled $a_{i}$ for
every node labelled $x_{i}$, reaches the terminal node labelled $f(\overline{a})$.
A branching program is called read-once (also known as a Free BDD) if in every path from the root to a terminal node each $x_{i}$ appears at most once.
\end{defn}
Allender et al. \cite{AKRR03} show that the minimal size of a branching programs for a given truth-table cannot be approximated up to a factor $N^{1-\varepsilon}$
in bounded polynomial probabilistic time, under the assumption that there is no algorithm with polynomial expected running time that factors Blum integers. Our results below use the same assumption to achieve similar results using different techniques (as in \cite{KW09}). These techniques were generalized in Theorem \ref{thm:PRFandTTM}, and the results below use that generalization. \\
\indent In \cite{MV00}, it is stated that the problem of truth-table size-minimization for read-once branching programs is open. An exact minimization algorithm for a read once branching programs exists in the literature \cite{GD99}, but it requires super-exponential time for certain functions, and therefore is not useful for our purposes.

A common restriction that is often applied over branching programs, is to limit the variables in any path to appear according to some fixed permutation. 

\begin{defn} \label{defn:OBDD}
An Ordered Binary Decision Diagram (OBDD) is a branching
program that can be divided into layers $L_{1},\cdots,L_{d+1}$, for $d\le n$, such
that all nodes in each layer except the last are labelled with the same variable.
$L_{1}$ is a singleton that contains the root and $L_{d+1}$ contains
the terminal nodes. Directed edges may exist
between layers $L_{i},L_{j}$ only if $i<j$.
\end{defn}

It is known that OBDDs have a polynomial truth-table minimization algorithm \cite{FS90}. We show that a subclass of OBDD has a faster truth table minimization algorithm, that uses a learning algorithm from \cite{VW93} as a black box.

\begin{defn} \label{defn:muBP}
A $\mu$-branching program is a branching program such that every variable appears at most once in the entire program. 
\end{defn}

Since every directed acyclic graph has a topological sort, every $\mu$-branching program may be seen as an OBDD of width 1. We shall see that the language of all truth-tables that have a corresponding $\mu$ branching program has a decision algorithm which is faster than using the algorithm of \cite{FS90} and accepting iff its output is an OBDD of width 1.

\section{Hardness Results}

\begin{thm} \label{thm:BP_approx}
Assuming that factoring Blum integers\footnote{a Blum integer is a number $n$ such that $n=pq$ where $p$ and $q$ are prime and congruent to $3$ modulo 4.} is not possible in probabilistic polynomial time, $BP_{size}(f)$
(i.e. the optimal branching program size for a function $f$) cannot be approximated within a factor of $2^{cn}$ for any $0<c<\frac{1}{2}$. Namely, there is no polynomial time algorithm $B$ such that 
\begin{eqnarray*}\frac{1}{2^{cn}}\cdot BP_{size}(f)\leq B(T_f)\leq 2^{cn}\cdot BP_{size}(f) \end{eqnarray*}
\end{thm}

Notice that the result of \cite{AKRR03} is stated with \textit{one-sided} error, i.e. that there is no polynomial algorithm $B$ such that $BP_{size}(f) \le B(T_f) \le 2^{cn} \cdot BP_{size}(f)$ for $0<c<1$. By multiplying with a proper factor, we have that both results are equivalent. We retain the two sided error terminology for convenience.

\begin{proof}
We show that the conditions 1,2 and 3 of Theorem \ref{thm:PRFandTTM} are met, where $C$ is the branching program model. For condition 1, a counting argument by \cite{W00} shows that all but exponentially small fraction of functions on $n$ inputs require branching programs of size at least ${2^{n}}\cdot n^{-1}(1-n^{-\frac{1}{2}})$. For condition 2, Naor and Reingold \cite{NR97} proved the existence of a pseudo-random function ensemble in $NC^{1}$ under the assumption of intractability of factoring Blum integers. We calculate the branching program complexity of the PRF from \cite{NR97} using the following proposition, provable by straightforward
induction.
\begin{prop}
For every boolean function $f:\{0,1\}^{n}\to \{0,1\}$, if there exists a circuit of depth at most $d$ computing $f$, then there is a branching program of size at most $2^{d}$ computing $f$.
\end{prop}
Therefore the branching program complexity of the PRF from \cite{NR97} is at most $n^{\beta}$ for some constant $\gb$, when $n$ is the number of variables. 

We now find $\e$ and $\ga(n)$ that satisfy condition 3. $\ga(n)$ must satisfy
\[
2^{\frac{\gb n}{\e}} \le \frac{1}{\ga^2 (n)} \cdot \frac{2^n}{n}\cdot \left(1-\frac{1}{\sqrt{n}}\right),
\]
i.e., 
\[
\ga(n) \le 2^{\frac{n}{2} \cdot \left( 1-\frac{\gb}{\e}\right)}\cdot \sqrt{\frac{1-\frac{1}{\sqrt{n}}}{n}}.
\]
For any $0<c<\frac{1}{2}$ we may choose 
\[
\e > \frac{\gb n}{(1-2c)n-\log n + \log \left(1-\frac{1}{\sqrt{n}}\right)},
\]
since the r.h.s of the above equation goes to $\frac{\gb}{1-2c}$ as $n$ goes to infinity. This choice of $\e$ gives us that
\[
2^{\frac{n}{2} \cdot \left( 1-\frac{\gb}{\e}\right)}\cdot \sqrt{\frac{1-\frac{1}{\sqrt{n}}}{n}} \ge 2^{cn}
\]
therefore we may choose $\ga(n) = 2^{cn}$, and the claim follows.
\end{proof}

\section{Efficient Algorithm} \label{sctn:BPEfficient}

If we wish to decide whether a given truth-table is representable by a $\mu$-branching program, we may use the algorithm of \cite{FS90} for truth-table minimization of OBDDs, and accept iff the resulting OBDD has width 1 (see Section \ref{sctn:BPperliminaries}). This would yield an algorithm with time complexity of $O(n^2 \cdot 3^n)$, since this is the worst case complexity of the algorithm by \cite{FS90}. The algorithm we present below slightly improves on that, by returning the correct answer in $O(n\cdot 2^n)$ time, using a learning algorithm by \cite{VW93} and applying Theorem \ref{thm:Learning}. 

\begin{thm} 
Given a full truth-table $T_f$ of a function $f$, there exists an algorithm that finds an equivalent $\mu$-branching program for $f$ if such a branching program exists, and rejects otherwise. The algorithm requires $O(n\cdot 2^n)$ time, when $n$ is the number of variables of $f$.
\end{thm}

\begin{proof}
We use a result by \cite{VW93} and Theorem \ref{thm:Learning}. \cite{VW93} provide a meta-algorithm for learning concept classes under certain restrictions, and later use it to learn $\mu$-branching programs efficiently. We refer the reader to Theorem 3 of \cite{VW93} for details about the meta-algorithm used to learn $\mu$-branching programs. As explained in \cite{VW93}, the algorithm produces hypotheses from inside the concept class, until the correct one is found. Therefore this learning algorithm is both proper and exact, and the first condition of Theorem \ref{thm:Learning} is met. Second, notice that any $\mu$-branching program for a function on $n$ variables that depends on all of them is of size exactly $n$. We therefore add a preprocessing phase to the minimization algorithm, in which we make sure that $f$ depends on all its variables. If so, we feed its entire truth-table to the learning algorithm. If not, we reduce the truth-table of $f$ to a truth-table of the sub-function of it that is equivalent to $f$ and does depend on all its variables. This allows us to use Theorem \ref{thm:Learning} to get a polynomial truth-table minimization algorithm, since the minimization algorithm $B$ stated in that theorem may just output its input. \\
\indent As stated in \cite{VW93}, their learning algorithm requires $O(n)$ equivalence queries, $poly(n)$ time and $poly(n)$ membership queries. Thus, as explained in Theorem \ref{thm:Learning}, after checking for the dependence of $f$ in all its variables (easily implementable in $O(n\cdot 2^n)$ time) we get a truth-table minimization that requires $O(n\cdot 2^n)$ time. 
\end{proof}

\bibliographystyle{alpha}
\bibliography{thesis}

\end{document}